\documentclass[a4paper,11pt]{article}
\input{Preamble.tex}

\begin{document}

\pagenumbering{Roman}

\hypersetup{pageanchor=false}

\newcommand{\acknow}{F.F. was supported by the Research Council of Norway via the project BWCA (grant no. 314528). T.K. was supported by the Research Council of Norway via the project BWCA (grant no. 314528), by Grant Numbers 16582 and 54451, Basic Algorithms Research Copenhagen (BARC), from the VILLUM Foundation, and by the European Union under Marie Skłodowska-Curie Actions (MSCA), project no. 101206430.}

\title{Finding sparse induced subgraphs on graphs of bounded induced matching treewidth\thanks{\acknow}}

\author{
    Hans L. Bodlaender \\
    \small{Utrecht University, the Netherlands} \\
    \small{\texttt{h.l.bodlaender@uu.nl}}
    \and
    Fedor V. Fomin \\
    \small{University of Bergen, Norway} \\
    \small{\texttt{fedor.fomin@uib.no}}
    \and
    Tuukka Korhonen \\
    \small{University of Copenhagen, Denmark} \\
    \small{\texttt{tuko@di.ku.dk}}
}

\maketitle

\thispagestyle{empty}

\begin{abstract}
The \emph{induced matching width} of a tree decomposition of a graph $G$ is the cardinality of a largest induced matching $M$ of $G$, such that there exists a bag that intersects every edge in~$M$.
The \emph{induced matching treewidth} of $G$, denoted by $\treemu(G)$, is the minimum induced matching width of a tree decomposition of $G$.
The parameter $\treemu$ was introduced by Yolov~[SODA~'18], who showed that, for example, \mwis can be solved in polynomial-time on graphs of bounded $\treemu$.
Lima, Milanič, Muršič, Okrasa, Rzążewski, and Štorgel [ESA~'24] conjectured that this algorithm can be generalized to a meta-problem called \mwisbt, where we are given a vertex-weighted graph $G$, an integer $w$, and a $\CMSO_2$-sentence $\Phi$, and are asked to find a maximum-weight set $X \subseteq V(G)$ so that $G[X]$ has treewidth at most $w$ and satisfies $\Phi$.
They proved the conjecture for some special cases, such as for the problem \mwif.

In this paper, we prove the general case of the conjecture.
In particular, we show that \mwisbt is polynomial-time solvable when $\treemu(G)$, $w$, and $|\Phi|$ are bounded.
The running time of our algorithm for $n$-vertex graphs $G$ with $\treemu(G) \le k$ is $f(k, w, |\Phi|) \cdot n^{\OO(k w^2)}$ for a computable function $f$.
\end{abstract}

\begin{textblock}{20}(-0.5, 4.8)
\includegraphics[width=100px]{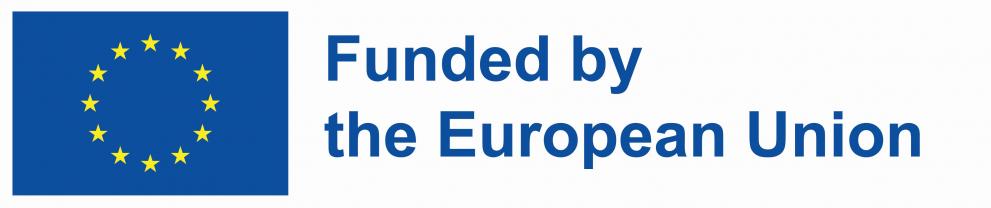}%
\end{textblock}

\thispagestyle{empty}

\newpage


\setcounter{page}{1}



\pagenumbering{arabic}

\hypersetup{pageanchor=true}

\clearpage
\setcounter{page}{1}

\section{Introduction}
Graph decompositions lie at the core of modern algorithmic graph theory.  
While treewidth is still the most influential structural graph parameter, many successful generalizations have emerged, including clique-width, rank-width, mim-width, Boolean-width, and tree-independence number~\cite{DBLP:journals/jcss/CourcelleER93,DBLP:journals/jct/OumS06,vatshelle:thesis,DBLP:journals/tcs/Bui-XuanTV11,DBLP:journals/jctb/DallardMS24,DBLP:conf/soda/Yolov18}.
Continuing this line of research, Yolov introduced the \emph{induced matching treewidth} of a graph $G$, denoted by $\treemu(G)$~\cite{DBLP:conf/soda/Yolov18}\footnote{Yolov called his parameter ``minor-matching hypertree width'', but we use ``induced matching treewidth'', following~\cite{DBLP:conf/esa/LimaMMORS24}}.
Whereas treewidth limits the \emph{size} of each bag, $\treemu$ limits the size of an induced matching that can intersect a bag.

More formally, for a graph $G$ and a set of vertices $X \subseteq V(G)$, we denote by $\mu(X)$ the cardinality of a largest induced matching $M$ so that for each edge $uv \in M$, it holds that $\{u,v\} \cap X \neq \emptyset$.
Then, for a tree decomposition $\Tc = (T,\bag)$, where $T$ is a tree and $\bag$ a bag-function $\bag \colon V(T) \to 2^{V(G)}$, the \emph{induced matching width} of $\Tc$ is $\mu(\Tc) = \max_{t \in V(T)} \mu(\bag(t))$.
The induced matching width of $G$, $\treemu(G)$, is the minimum of $\mu(\Tc)$ over all tree decompositions $\Tc$ of $G$.
Note that $\treemu(G) \le \tw(G)+1$, where $\tw(G)$ denotes the treewidth of $G$.
Furthermore, the induced matching treewidth of a complete graph or a complete bipartite graph is at most~$1$, while they have unbounded treewidth.

When working with graph decompositions, two fundamental algorithmic issues arise:

\begin{enumerate}
  \item  How can one compute a “good’’ decomposition---typically, one whose width is close to optimal---in the fastest possible time?
  \item  Which hard computational problems become tractable once such a low-width decomposition is available?
\end{enumerate}

For example, while computing the treewidth of a graph is NP-hard  \cite{ArnborgCP87}, a rich toolbox of FPT~\cite{Bodlaender96}, approximation~\cite{FeigeHL08}, and FPT-approximation~\cite{Korhonen21} algorithms is available for this purpose. The second question is, in a sense, even more rewarding: Once a tree decomposition of width $k$ is given, a large family of problems can be solved in time $f(k)\cdot n$, by dynamic programming.  Courcelle’s celebrated meta-theorem \cite{DBLP:journals/iandc/Courcelle90} crystallizes this: Every graph property expressible in counting monadic second-order logic ($\CMSO_2$) is decidable in linear time on graphs of bounded treewidth.

Similar “two-question’’ narratives play out for any width measures, including $\treemu$. 
For computing tree decompositions with small induced matching width, Yolov~\cite{DBLP:conf/soda/Yolov18} gave an algorithm that for a given $n$-vertex graph $G$ and integer $k$, in time $n^{\OO(k^3)}$ either constructs a tree decomposition of $G$ whose induced matching width is $\OO(k^3)$ or correctly reports that $\treemu(G)$ is larger than~$k$.
Dallard, Fomin, Golovach, Korhonen, and Milanic~\cite{DBLP:conf/icalp/DallardFGKM24} improved the running time and approximation ratio of  Yolov's algorithm by giving an algorithm of running time  $2^{\OO(k^2)} n^{\OO(k)}$. Their algorithm either outputs a tree decomposition of $G$ of $\mu$-width at most $8k$, or concludes that $\treemu(G)$  is larger than~$k$.
Dallard et al. also showed that for every constant $k\geq 4$, it is NP-complete to decide whether $\treemu(G)\leq k$ for a given graph~$G$.\footnote{The results of Dallard et al. \cite{DBLP:conf/icalp/DallardFGKM24} are stated for related parameter tree-independence number, but due to a reduction by Yolov~\cite{DBLP:conf/soda/Yolov18} between computing it and induced matching treewidth, they hold for $\treemu$ as well.}

With respect to the second question, the landscape for $\treemu$ is far from settled, and identifying the limits of algorithmic applicability of  induced matching treewidth is the main motivation of our work.  Yolov showed that graphs of bounded $\treemu$ admit polynomial-time algorithms for \textsc{Maximum Independent Set}, \textsc{$c$-Coloring} for constant $c$, and, more generally, for \textsc{Graph Homomorphism} to a fixed pattern graph~\cite{DBLP:conf/soda/Yolov18}.
Lima, Milanič, Muršič, Okrasa, Rzążewski, and Štorgel recently extended Yolov’s results by proving that \textsc{Minimum-Weight Feedback Vertex Set} (equivalently, \textsc{Maximum-Weight Induced Forest}) can be solved in time $n^{\mathcal{O}(k)}$ on graphs $G$ with $\treemu(G) \le k$~\cite{DBLP:conf/esa/LimaMMORS24}.

More satisfactory than designing individual dynamic programming algorithms for graphs of bounded $\treemu$ would be a general meta-theorem, akin to Courcelle's theorem~\cite{DBLP:journals/iandc/Courcelle90}, providing an automated way of obtaining algorithms for graphs of bounded $\treemu$.
Inspired by an earlier meta-theorem for graphs with a polynomial number of minimal separators by Fomin, Todinca, and Villanger~\cite{FominTV15}, Lima et al.~\cite{DBLP:conf/esa/LimaMMORS24} conjectured that the following meta-problem, called \mwisbt, is solvable in polynomial-time on graphs of bounded $\treemu$ (when the inputs $w$ and $\Phi$ are also bounded).

\medskip

\defproblema{ \mwisbt}{Graph $G$, a weight function $\we \colon V(G) \to \R$, an integer $w$, and a $\CMSO_2$-sentence $\Phi$. }{Find a maximum-weight set $X \subseteq V(G)$ such that $G[X]$ has treewidth at most $w$ and satisfies $\Phi$.}

This problem is equivalent to \mwis when $w = 0$ and $\Phi$ is always true, and to \textsc{Maximum-Weight Induced Forest} when $w = 1$ and $\Phi$ is always true.
Furthermore, it captures, for example, the problem \textsc{Weighted Longest Induced Cycle} when $w=2$ and $\Phi$ describes that $G[X]$ is a cycle.
In addition to proving the \textsc{Maximum-Weight Induced Forest} case of the conjecture, Lima et al. proved the conjecture also when $\treemu$ is replaced by a more restrictive parameter \emph{tree-independence number}~\cite{DBLP:conf/esa/LimaMMORS24}.
Tree-independence number, $\treealpha$, is defined like $\treemu$, but using the function $\alpha(\bag(t))$, the cardinality of a largest independent set in $\bag(t)$, instead of the function $\mu(\bag(t))$.
Note that $\treemu(G) \le \treealpha(G)$, but, for example, complete bipartite graphs have unbounded $\treealpha$ but bounded $\treemu$.

\subsection{Our contribution.}
We resolve the conjecture of Lima et al.~\cite{DBLP:conf/esa/LimaMMORS24} in the affirmative.

\begin{restatable}{theorem}{maintheorem}
\label{the:main}
There is an algorithm that, given as input an $n$-vertex graph $G$, a weight function $\we \colon V(G) \to \R$, integers $k$ and $w$, and a $\CMSO_2$-sentence $\Phi$, in time $f(k, w, |\Phi|) \cdot n^{\OO(k w^2)}$, where $f$ is a computable function, either (1) returns a maximum-weight set $X \subseteq V(G)$ so that $G[X]$ has treewidth at most $ w$ and satisfies $\Phi$, or (2) determines that no such set $X$ exists, or (3) concludes that $\treemu(G) > k$.
\end{restatable}

When $w$ and $\Phi$ are constants, the running time is $f(k) \cdot n^{\OO(k)}$, which matches (when $n$ is sufficiently large compared to $k$) to the running times known for \mwis and \textsc{Maximum-Weight Induced Forest}~\cite{DBLP:conf/soda/Yolov18,DBLP:conf/esa/LimaMMORS24}.
As for the computational model for the real weights, we only use the assumptions that two numbers can be added and compared efficiently.
We note that \Cref{the:main} supports also negative weights, although we are not aware of any application where they would turn up naturally (i.e., a minimization problem).

The class of optimization problems addressed by \Cref{the:main} is  broad.  
It subsumes the weighted variants of several well-studied problems, including
\textsc{Maximum Independent Set},
\textsc{Maximum Induced Matching},
\textsc{Longest Induced Path},
\textsc{Longest Induced Cycle},
\textsc{Maximum Induced Forest}, and
\textsc{Maximum Induced Subgraph Excluding a Planar Minor}.  
For more examples of problems expressible by this framework, we refer the reader to~\cite{FominTV15}.

Even more generally, by using a known fact that $K_{t,t}$-subgraph-free graphs of bounded $\treemu$ have bounded treewidth~\cite{DBLP:conf/esa/LimaMMORS24}, we can replace the bounded treewidth condition in \Cref{the:main} by $K_{t,t}$-subgraph-freeness.
A graph is $K_{t,t}$-subgraph-free if it does not contain the complete bipartite graph $K_{t,t}$ as a subgraph.

\begin{restatable}{corollary}{corkttfree}
\label{cor:kttfreealg}
There is an algorithm that, given as input an $n$-vertex graph $G$, a weight function $\we \colon V(G) \to \R$, integers $k$ and $t$, and a $\CMSO_2$-sentence $\Phi$, in time $f(k, t, |\Phi|) \cdot n^{g(k,t)}$, where $f$ and $g$ are computable functions, either (1) returns a maximum-weight set $X \subseteq V(G)$ so that $G[X]$ is $K_{t,t}$-subgraph-free and satisfies $\Phi$, or (2) determines that no such set $X$ exists, or (3) concludes that $\treemu(G) > k$.
\end{restatable}

Note that $K_{t,t}$-subgraph-freeness (equivalently, excluding a clique and a biclique as induced subgraphs) is the most general hereditary condition that in conjunction with bounded $\treemu$ guarantees bounded treewidth, as both cliques and bicliques have bounded induced matching treewidth but unbounded treewidth.
The class of $K_{t,t}$-subgraph-free graphs includes planar graphs (which are $K_{3,3}$-free due to Kuratowski's theorem), and more generally, $H$-minor-free graphs, for every fixed graph $H$.
It also contains graphs of bounded maximum degree, as well as graphs of bounded degeneracy.
Thus, \Cref{cor:kttfreealg} allows us to directly capture problems such as \textsc{Maximum Planar Induced Subgraph} (equivalently, \textsc{Minimum Planar Vertex Deletion}), \textsc{Maximum $\mathcal{F}$-Minor-Free Induced Subgraph} (equivalently, \textsc{Minimum $\mathcal{F}$-Minor-Free Deletion}), and \textsc{Maximum Cubic Induced Subgraph}.

We believe that our techniques allow capturing even more problems than captured by \Cref{the:main} and \Cref{cor:kttfreealg}, and we will discuss this and other future directions in the Conclusions section at the end of the paper (\Cref{sec:conclusions}).

\paragraph{Organization of the paper.}
We start by presenting the sketch of the proof of \Cref{the:main} in \Cref{section:overview}.
Then, in \Cref{sec:preliminaries} we provide definitions and preliminary results on graphs, treewidth, and logic.
The proofs of \Cref{the:main} and \Cref{cor:kttfreealg} will be presented in full detail in \Cref{sec:bigsignatures,sec:innertree,sec:dynprog,sec:finalproof}, with \Cref{sec:bigsignatures,sec:innertree,sec:dynprog} each presenting one of the three parts of the algorithm, and \Cref{sec:finalproof} putting them together.
We conclude the paper with \Cref{sec:conclusions}.

\section{Overview}\label{section:overview}
In this section we present a sketch of the proof of \Cref{the:main}.
We first explain the high-level approach in \Cref{subsec:overview:highlevel} and state two key lemmas, and then in \Cref{subsec:overview:sketchyolovgen,subsec:overview:sketchinner} sketch the proofs of these key lemmas.

\subsection{High-level plan}
\label{subsec:overview:highlevel}
The algorithm of \Cref{the:main} works by dynamic programming on a given tree decomposition $\Tc$ with small $\mu(\Tc)$.
By previous work on computing such tree decompositions~\cite{DBLP:conf/soda/Yolov18,DBLP:conf/icalp/DallardFGKM24}, we can assume that we are given such a tree decomposition whose width is optimal within a constant factor.

\paragraph{Yolov's algorithm.}
We start by reviewing the dynamic programming algorithm of Yolov~\cite{DBLP:conf/soda/Yolov18} for \mwis parameterized by $\mu(\Tc)$.
A \emph{maximal independent set} in a graph $G$ is an inclusion-wise maximal set of vertices that is an independent set.
The core lemma that Yolov proves is the following.

\begin{lemma}[\cite{DBLP:conf/soda/Yolov18}]
\label{lem:overview:yolovlem}
Let $G$ be an $n$-vertex graph and $B \subseteq V(G)$ a set of vertices.
There are at most $n^{\OO(\mu(B))}$ different sets $I \cap B$, where $I$ is a maximal independent set of $G$, and they can be enumerated in $n^{\OO(\mu(B))}$ time.
\end{lemma}

The proof of \Cref{lem:overview:yolovlem} is short but not very easy.
It is a modification of an earlier proof establishing that the number of maximal independent sets of a graph $G$ is at most $n^{\OO(\mu(G))}$~\cite{DBLP:journals/networks/BalasY89}.

By \Cref{lem:overview:yolovlem},  for each bag $B = \bag(t)$ of the input tree decomposition $\Tc = (T,\bag)$, we can in $n^{\OO(\mu(\Tc))}$ time enumerate a family $\fami(t)$ of subsets of $\bag(t)$, such that  $I \cap \bag(t) \in \fami(t)$ for any optimal solution $I$.
Here, we use that every optimal solution $I$ is a maximal independent set, which can be assumed after deleting all vertices with weight $\le 0$.

Now, performing dynamic programming to find an optimal solution is not difficult.
Assume that $\Tc = (T,\bag)$ is rooted, and denote by $\subtree(t) \subseteq V(G)$ the union of the bags of the subtree of $\Tc$ rooted at $t$.
In a state of the dynamic programming we store, for each pair $t \in V(T)$ and $S \in \fami(t)$, a maximum weight independent set $I$ of $G$ such that  (1) $I \subseteq \subtree(t)$, (2) $I \cap \bag(t') \in \fami(t')$ for all descendants $t'$ of $t$, and (3) $I \cap \bag(t) = S$.
Assuming that $T$ is a binary tree, the transitions to compute all of the states are straightforward.
Since there is an optimal solution $I$ satisfying $I \cap \bag(t) \in \fami(t)$ for all $t$, this dynamic programming finds an optimal solution in $n^{\OO(\mu(\Tc))}$ time.

\paragraph{Analogue of Yolov's lemma.}
The first ingredient of our algorithm is a statement similar to \Cref{lem:overview:yolovlem}.
Let $(G, \we, \Phi, w)$ be an instance of \mwisbt, and let us say that $X \subseteq V(G)$ is \emph{the optimal solution} to $(G,\we, \Phi, w)$ if $X$ is a maximum-weight solution to $(G, \we, \Phi, w)$, and furthermore, among maximum-weight solutions it is the lexicographically  largest (subject to some fixed vertex ordering of $G$).
We show the following lemma.

\begin{restatable}[Informal version of \Cref{lem:bagsignatures}]{lemma}{lemoverviewyolovgen}
\label{lem:overview:yolovgen}
Given an instance $(G, \we, \Phi, w)$ of \mwisbt, and a set of vertices $B \subseteq V(G)$ with $\mu(B) \le k$, we can in $f(k, w, |\Phi|) \cdot n^{\OO(kw)}$ time, where $f$ is a computable function, enumerate a family $\fami$ of subsets of $B$ such that  if $X$ is the optimal solution to $(G, \we, \Phi, w)$, then $X \cap B \in \fami$.
\end{restatable}

In particular, if $(G, \we, \Phi, w)$ describes the \mwis problem with positive vertex-weights, then \Cref{lem:overview:yolovgen} follows from \Cref{lem:overview:yolovlem}.
The proof of \Cref{lem:overview:yolovgen} builds on ideas from~\cite{DBLP:conf/esa/LimaMMORS24}, where a similar lemma was proven for \mwif, but requires also significant new ideas.

For \mwis, the rest of the algorithm after \Cref{lem:overview:yolovlem} was easy.
This is not the case for \mwisbt, and was not the case even for \mwif.
In particular, for \mwif, one needs to check during the dynamic programming that we are indeed building an acyclic induced subgraph.
This is done in~\cite{DBLP:conf/esa/LimaMMORS24} by tracking the connected components of the partial solution restricted to the current bag, analogously to the dynamic programming for \mwif parameterized by treewidth.
As the intersection of the partial solution with the bag can have unbounded size, this requires an additional observation that all but a bounded number of vertices in the intersection must have degree at most $1$ in the solution.

For \mwisbt, such ad hoc ideas will not work, as we have to capture the full expressivity of $\CMSO_2$ on the partial solution that we are building, with our only tool in hand being that the solution has bounded treewidth.
This leads to the next main ingredient of our algorithm.

\paragraph{Inner tree decomposition.}
The most commonly used approach for testing if a graph of bounded treewidth satisfies a $\CMSO_2$ formula is dynamic programming over a tree decomposition of bounded width.
(It is also possible to test such properties with help of graph reduction~\cite{ArnborgCPS93}, but we will not build upon that here.)
Therefore, a natural idea is to define an ``inner tree decomposition'' for the solution $X$ inside the tree decomposition $\Tc$, and evaluate the dynamic programming for $\CMSO_2$ as a part of the overall dynamic programming scheme.
In particular, given a tree decomposition $\Tc = (T,\bag)$ of $G$ and a hypothetical solution $X \subseteq V(G)$ with $\tw(G[X]) \le w$, we could consider the tree decomposition $\Tc_X = (T, \bag_X)$, with $\bag_X(t) = \bag(t) \cap X$, of $G[X]$, and evaluate the dynamic programming for $\CMSO_2$ on $G[X]$ along $\Tc_X$ at the same time as finding $X$ itself.
This approach was indeed used in the context of the parameter ``tree-independence number'' by~\cite{DBLP:conf/esa/LimaMMORS24}.
However, it clearly fails in the context of $\treemu$, as the size of $\bag_X(t) = \bag(t) \cap X$ can be unbounded even when $X$ is an independent set.

Nevertheless, we manage to make the idea work with a more complicated definition of $\Tc_X$.
Our definition hits the sweet spot in that it is simple enough to consider the tree decomposition $\Tc_X$ on top of doing dynamic programming on $\Tc$, while still achieving that $\Tc_X$ has bounded width.
In particular, we prove the following lemma.

A \emph{nice tree decomposition} is a rooted tree decomposition where every node $t$ either (1) has two children $c_1$, $c_2$ with $\bag(c_1) = \bag(c_2) = \bag(t)$, (2) has one child $c$ with $|\bag(c) \setminus \bag(t)| + |\bag(t) \setminus \bag(c)| \le 1$, or (3) has no children and $\bag(t) = \emptyset$.

\begin{restatable}[Informal version of \Cref{lem:virtdecomp}]{lemma}{lemoverviewinnerdecomp}
\label{lem:overview:innerdecomp}
Let $G$ be a graph, $\Tc = (T,\bag)$ a nice tree decomposition of $G$ with $\mu(\Tc) \le k$, and $X \subseteq V(G)$ a set with $\tw(G[X]) \le w$.
There exists a partition $(X_1, X_2)$ of $X$ and a bag-function $\bag_X \colon V(T) \to 2^{X}$, such that 
\begin{itemize}
\item $\Tc_X = (T, \bag_X)$ is a tree decomposition of $G[X]$ of width $\OO(k w^2)$,
\item $\bag(t) \cap X_1 \subseteq \bag_X(t) \cap X_1 \subseteq \subtree_{\Tc}(t) \cap X_1$ for all $t \in V(T)$,
\item $X_2$ is an independent set of $G$, and
\item for each $v \in X_2$, the only node $t \in V(T)$ with $v \in \bag_X(t)$ is the closest node $t$ to the root of $T$ with $v \in \bag(t)$.
\end{itemize}
\end{restatable}

Furthermore, the actual statement of \Cref{lem:overview:innerdecomp} as \Cref{lem:virtdecomp} provides several further properties, assuming that $\Tc$ is a ``supernice'' tree decomposition.

Armed with \Cref{lem:overview:yolovgen,lem:overview:innerdecomp}, the dynamic programming algorithm for solving \mwisbt is rather straightforward but very tedious due to numerous cases and a lot of bookkeeping.
The main idea of this part of the paper is the design of the inner tree decomposition $\Tc_X = (T,\bag_X)$ in \Cref{lem:overview:innerdecomp}.
In particular, it is essential to choose $\Tc_X = (T,\bag_X)$ to have structure suitable for understanding it in the middle of dynamic programming while constructing $X$, but at the same time we have to keep its width bounded to evaluate $\CMSO_2$ on $G[X]$.

The actual dynamic programming algorithm keeps track of (1) the intersection $X \cap \bag(t)$ of the partial solution $X$ with $\bag(t)$, which must be in $\fami(t)$ given by \Cref{lem:overview:yolovgen}, (2) the partition of $X \cap \bag(t)$ to $X_1 \cap \bag(t)$ and $X_2 \cap \bag(t)$, which is easy to track as $|X_1 \cap \bag(t)| \le |\bag_X(t)| \le \OO(k w^2)$, and (3) the state of the dynamic programming that evaluates whether $G[X]$ satisfies the $\CMSO_2$-formula $\Phi$ using the inner tree decomposition $\Tc_X$.
We need to pay attention to make sure that the tree decomposition $\Tc_X$ is indeed a tree decomposition of $G[X]$, for which it suffices to make sure that (1) $X_2$ is indeed an independent set and (2) we do not ``forget'' vertices in $X_1$ before their neighbors in $X_2$ have made their appearance in $\Tc_X$.

\subsection{Proof sketch of \Cref{lem:overview:yolovgen}}
\label{subsec:overview:sketchyolovgen}
Let us re-state \Cref{lem:overview:yolovgen} and sketch its proof.

\lemoverviewyolovgen*
\begin{proof}[Proof sketch]
The general approach to proving \Cref{lem:overview:yolovgen} is to define a ``signature'' that uniquely defines $X \cap B$, such that  there are at most $f(k, w, |\Phi|) \cdot n^{\OO(kw)}$ different signatures.
In the first part of this proof we define a ``basic signature'' that only uses the fact that $G[X]$ has bounded treewidth (but not the optimality $X$), which defines all vertices in $X \cap B$ except some well-structured ``dangling vertices''.
In the second part, we use $\Phi$ and the optimality of $X$ to also include the information about the dangling vertices in the signature.

The first step towards the proof of \Cref{lem:overview:yolovgen} is to observe that the set of edges of $G[X]$ incident to $B$ must have bounded vertex cover.
Indeed, if the graph $G[X]^B$ consisting of the edges of $G[X]$ incident to $B$ would have vertex cover number $> 2 \cdot k \cdot (w+1)$, then it would contain a matching $M$ of size $> k \cdot (w+1)$, which by $\tw(G[X]) \le w$ would contain a subset $M' \subseteq M$ of size $|M'|>k$ that is an induced matching, which would contradict $\mu(B) \le k$.
Therefore, the first part of the signature is a set $C \subseteq X$ of size $|C| \le \OO(kw)$ that is a vertex cover of $G[X]^B$.
Note that vertices of $C$ can be also outside of $B$.

The vertices $X \cap B \setminus C$ then form an independent set.
The second part of the signature is a maximal independent set $S$ of $G[B]$ such that  $X \cap B \setminus C \subseteq S$.
By \Cref{lem:overview:yolovlem}, there are at most $n^{\OO(k)}$ such maximal independent sets $S$.

The knowledge of $C$ and $S$ determines that $X \cap B$ is a superset of $C \cap B$ and a subset of $S \cup (C \cap B)$.
An obvious reason for why a particular vertex in $S$ should not be in $X \cap B$ is that its inclusion to $X$ would cause $C$ to not be a vertex cover of $G[X]^B$.
Therefore, we need to fix some vertices in $X \setminus C$ to ``kick out'' their neighbors from $S$.
Let $S_D \subseteq S$ be the vertices in $S$ that have a neighbor in $X \setminus C$, and thus cannot be in $X$, in particular, $S_D = S \cap N(X \setminus C)$.
It turns out, we can choose a subset $D \subseteq X \setminus C$ of size $|D| \le \OO(kw)$ such that  $S_D = S \cap N(D)$.
This is proven by an argument that uses $\tw(G[X]) \le w$ and $\mu(B) \le k$.

Now, the triple $(C, S, D)$ ``tells'' us that $X \cap B$ is a superset of $C \cap B$ and a subset of $(S \setminus N(D)) \cup (C \cap B)$.
However, it does not tell whether the vertices in $S \setminus (N(D) \cup C)$ are in $X \cap B$ or not.
We call such vertices \emph{dangling vertices}.
We observe that for each dangling vertex $v$ it holds that $N(v) \cap X \subseteq C$.
It remains to encode the information about the dangling vertices in the signature.
For this, we use $\Phi$ and the optimality of $X$.

We group dangling vertices $v$ based on the neighborhood $N(v) \cap X \subseteq C$ into at most $2^{|C|} \le 2^{\OO(kw)}$ groups, and encode each group independently.
We also further divide the groups based on whether the weight $\we(v)$ is negative or non-negative.
We observe that as each vertex in a group has the same neighborhood into $X$, for the satisfaction of $\Phi$ it only matters how many vertices we choose from the group, as different ways of choosing the same number of vertices yield isomorphic graphs $G[X]$ in the end.
Therefore, the optimal solution must choose some number of the highest-weight vertices from the group.
Another observation is that by using the finite-stateness of $\CMSO_2$ on graphs of bounded treewidth, we can show that the optimal solution should choose either a bounded number of vertices from a group, or all but at most a bounded number of vertices.
Here, ``bounded number'' is some bound depending on a computable function of $|\Phi|$ and $w$.
Therefore, for each group, we can encode which vertices are chosen by giving only one bit telling whether we choose almost none or almost all, and one number that is bounded by a function of $|\Phi|$ and $w$.
In total, we manage to encode all dangling vertices by an amount of information depending only on $k$, $|\Phi|$, and $w$.
\end{proof}

\subsection{Proof sketch of \Cref{lem:overview:innerdecomp}}
\label{subsec:overview:sketchinner}
Let us re-state \Cref{lem:overview:innerdecomp} and sketch its proof.

\lemoverviewinnerdecomp*
\begin{proof}[Proof sketch]
First, we classify the vertices in $X$ to high-degree and low-degree vertices.
We let $X_H \subseteq X$ consist of the vertices in $X$ that have degree $>2 (w+1)$ in $G[X]$, and $X_L = X \setminus X_H$ the other vertices.
Then, we obtain the partition $(X_1,X_2)$ of $X$ by letting $X_1$ consist of $X_H$ and the vertices in $X_L$ that have neighbors in $X_L$, and $X_2$ consist of the vertices in $X_L$ that have no neighbors in $X_L$.

After fixing $X_1$, $X_2$, and the tree decomposition $\Tc = (T,\bag)$, the tree decomposition $\Tc_X = (T,\bag_X)$ is essentially uniquely defined by the lemma statement.
Denote by $\topnode_{\Tc}(v) \in V(T)$ the node $t$ closest to the root with $v \in \bag(t)$.
For each vertex $v \in X_2$, the unique $t$ with $v \in \bag_X(t)$ is defined to be $\topnode_{\Tc}(v)$.

By the lemma statement, $\bag(t) \cap X_1 \subseteq \bag_X(t) \cap X_1$, so each vertex in $X_1$ is at least in the same bags in $\Tc_X$ as in $\Tc$.
This ensures that the edge-condition between vertices in $X_1$ is trivially satisfied.
However, to satisfy the edge-condition between vertices in $X_1$ and $X_2$, some of the vertices in $X_1$ need to be inserted to additional bags.
In particular, consider $u \in X_1$ and $v \in X_2$ with $uv \in E(G)$.
If $u \in \bag(\topnode_{\Tc}(v))$, then we do not need to do anything, as $u,v \in \bag_X(\topnode_{\Tc}(v))$ must hold by the definitions.
Otherwise, it holds that $\topnode_{\Tc}(u)$ is a descendant of $\topnode_{\Tc}(v)$, so we need to insert $u$ to all bags on the path between $\topnode_{\Tc}(u)$ and $\topnode_{\Tc}(v)$, all of which are ancestors of $\topnode_{\Tc}(u)$.
Doing this for all such pairs $u,v$, we end up ``extending'' each vertex $u \in X_1$ up to some ancestor of $\topnode_{\Tc}(u)$.

It remains to analyze the width of $\Tc_X$.
The fact that $\Tc$ is nice implies that $|\bag_X(t) \cap X_2| \le 1$, so we are interested in $|\bag_X(t) \cap X_1|$.
We consider three cases of vertices in $\bag_X(t) \cap X_1$: Those in $X_H \cap \bag(t)$, those in $X_L \cap X_1 \cap \bag(t)$, and those inserted because they have neighbor in $X_2 \cap \bag(a)$, where $a$ is an ancestor of $t$.

We start by showing $|X_H \cap \bag(t)| \le \OO(k w^2)$.
Because $\tw(G[X])\le w$, there is a subset $X_H' \subseteq X_H \cap \bag(t)$ of size $|X_H'| \ge |X_H \cap \bag(t)|/(w+1)$ that is an independent set.
By using Hall's theorem with the fact that each vertex in $X_H'$ has degree $> 2(w+1)$ in $G[X]$, we obtain either a matching of size $|X_H'|$ between $X_H'$ and $N(X_H') \cap X$, or a subgraph of $G[X]$ with average degree $>w+1$.
The latter case would imply $\tw(G[X]) > w$, so we indeed get the matching.
However, by using that $\tw(G[X]) \le w$, for any matching $M$ we can find a subset $M' \subseteq M$ with $|M'| \ge |M|/(w+1)$ that is an induced matching.
Therefore, $|X_H'| > k \cdot (w+1)$ would imply that $\mu(\bag(t)) > k$, so it must be that $|X_H'| \le k \cdot (w+1)$ and thus $|X_H \cap \bag(t)| \le k \cdot (w+1)^2$.

We then show $|X_L \cap X_1 \cap \bag(t)| \le \OO(k w^2)$.
Each vertex in $X_L \cap X_1 \cap \bag(t)$ has a neighbor in $X_L$, but all vertices in $X_L$ have degree at most $2 (w+1)$ in $G[X]$.
It follows that we can find a matching $M$ of size $|M| \ge |X_L \cap X_1 \cap \bag(t)|/(2 (w+1)+1)$ in $G[X_L]$ with each of its edges incident to $X_L \cap X_1 \cap \bag(t)$.
By again turning it into an induced matching, we conclude that $|X_L \cap X_1 \cap \bag(t)| \le k \cdot (2 (w+1)+1) \cdot (w+1)$.

We then consider the vertices $X_1 \cap \bag_X(t) \setminus \bag(t)$, that were inserted to fix the edge condition between $X_1$ and $X_2$.
Each such vertex must have a neighbor in $\bag(t) \cap X_2$.
So it suffices to bound $|N(\bag(t) \cap X_2)| \le \OO(k w^2)$. Since every vertex from  $\bag(t) \cap X_2$ is of degree at most
 most $2 (w+1)$, we obtain a matching $M$ of size $|M| \ge |N(\bag(t) \cap X_2)|/(2 (w+1))$ between $\bag(t) \cap X_2$ and $N(\bag(t) \cap X_2)$.
Again, by turning it into an induced matching, we conclude that $|N(\bag(t) \cap X_2)| \le k \cdot (2 (w+1)) \cdot (w+1)$.
\end{proof}

\section{Preliminaries}\label{sec:preliminaries}

In this section, we introduce basic and advanced definitions, as well as some fundamental graph-theoretic lemmas.

\subsection{Basic definitions}
For two integers $a,b$, we denote by $[a,b]$ the set of integers $i$ with $a \le i \le b$.
We use $[a]$ as a shorthand for $[1,a]$.

A graph consists of a set of vertices $V(G)$ and a set of edges $E(G)$.
For a set $X \subseteq V(G)$, we denote by $G[X]$ the subgraph of $G$ induced by $X$.
The set of neighbors of a vertex $v$ is denoted by $N(v)$ and the closed neighborhood by $N[v] = \{v\} \cup N(v)$.
For a set of vertices $X$, we denote $N(X) = \bigcup_{v \in X} N(v) \setminus X$ and $N[X] = N(X) \cup X$.

A \emph{tree} is a connected acyclic graph.
The vertices of a tree are sometimes called \emph{nodes}.
A \emph{rooted tree} is a tree where one node is selected as the root.
A \emph{binary tree} is a rooted tree where each node has at most two children.
A node $t$ of a rooted tree $T$ is an \emph{ancestor} of a node $s$ if $t$ is on the unique path from $s$ to the root.
In that case, $s$ is a \emph{descendant} of $t$.
Note that every node is both a descendant and an ancestor of itself.

A \emph{tree decomposition} of a graph $G$ is a pair $\Tc = (T,\bag)$, so that
\begin{enumerate}
\item $V(G) = \bigcup_{t \in V(T)} \bag(t)$, (vertex-condition)
\item for each $uv \in E(G)$, there is $t \in V(T)$ so that $\{u,v\} \subseteq \bag(t)$, (edge-condition) and 
\item for each $v \in V(G)$, the set of nodes $\{t \in V(T) \mid v \in \bag(t)\}$ forms a connected subtree of $T$ (connectedness-condition). 
\end{enumerate}

A \emph{rooted/binary tree decomposition} is a tree decomposition where the tree $T$ is rooted/binary.
If $\Tc$ is a rooted tree decomposition and $v \in V(G)$, we denote by $\topnode_{\Tc}(v) \in V(T)$ the node of $T$ so that $v \in \bag(\topnode_{\Tc}(v))$ and subject to that, $\topnode_{\Tc}(v)$ is as close to the root as possible. For node $t\in V(T)$, we denote by $\subtree_{\Tc}(t)$ the set of vertices of $G$ belonging to the bags corresponding to the nodes of the subtree of $T$ rooted at $t$. 

The \emph{width} of a tree decomposition is $\max_{v \in V(T)} |\bag(t)|-1$, and the \emph{treewidth} of a graph $G$, $\tw(G)$, is the minimum width of a tree decomposition of it.

A \emph{matching} is a set of edges $M \subseteq E(G)$ so that no vertex is incident to more than one edge in $M$.
An \emph{induced matching} is a matching $M$, so that the edges $M$ are the only edges of the subgraph induced by the vertices incident to $M$.
For a set $X \subseteq V(G)$, we denote by $\mu(X)$ the maximum cardinality of an induced matching $M$, so that each $e \in M$ is incident to a vertex in $X$.
We also denote $\mu(G) = \mu(V(G))$.
The \emph{induced matching width} of a tree decomposition is $\mu(\Tc) = \max_{t \in V(T)} \mu(\bag(t))$, and the \emph{induced matching treewidth} of a graph $G$, denoted by $\treemu(G)$, is the minimum induced matching width of a tree decomposition of it.


\subsection{$\CMSO_2$}
The Counting Monadic Second Order Logic ($\CMSO_2$) is a logic in which many graph properties can be expressed.
It has tight and well-studied connections to treewidth~\cite{DBLP:journals/iandc/Courcelle90,DBLP:journals/apal/Seese91,DBLP:conf/lics/BojanczykP16}.

In $\CMSO_2$ the variables are individual vertices, individual edges, sets of vertices, and sets of edges.
It includes binary predicates for testing if a vertex is incident to an edge, for testing if a vertex is in a set of vertices, for testing if an edge is in a set of edges, and for testing equality.
Furthermore, it includes for any two integers $q$ and $r$ with $0 \le q < r$ a unary predicate that tests whether the cardinality of a set is $q$ modulo $r$.
It includes also the standard logical connectives $\wedge$, $\vee$, and $\lnot$, and the quantifiers $\forall$ and $\exists$ which can bind variables.
A $\CMSO_2$-sentence $\Phi$ is a formula without free variables that can be formed using this vocabulary.
The \emph{length} $|\Phi|$ of a $\CMSO_2$-sentence $\Phi$ is the number of symbols plus the sum of the integers $q$ and $r$ used in the modular counting predicates.
We refer to the book of Courcelle and Engelfriet~\cite{CEbook} for a more detailed introduction to $\CMSO_2$.

We note that we can also verify $\tw(G) \le w$ in $\CMSO_2$.

\begin{lemma}[\cite{DBLP:conf/icalp/LagergrenA91}]
\label{lem:cmsotwcheck}
There is an algorithm that, given an integer $w$, constructs a $\CMSO_2$-sentence $\Phi_w$ so that a graph $G$ satisfies $\Phi_w$ if and only if $\tw(G) \le w$.
\end{lemma}

We did not state any bound for the running time of the algorithm or for $|\Phi_w|$, but by definition, they are bounded by computable functions of $w$.



\subsection{Tree decomposition automata}
In this paper, we do not directly manipulate $\CMSO_2$ in the dynamic programming, but instead use known results on \emph{tree decomposition automata} for verifying whether a graph satisfies a $\CMSO_2$-sentence.
The use of tree decomposition automata for formalizing dynamic programming on tree decompositions is a classic idea, especially in the context of $\CMSO_2$~\cite{DBLP:journals/jal/ArnborgLS91,DBLP:journals/iandc/Courcelle90,CEbook,DBLP:series/txcs/DowneyF13}.
We introduce our own definition of tree decomposition automaton that is convenient in the context of this paper, but note that the results from~\cite{DBLP:journals/jal/ArnborgLS91,DBLP:journals/iandc/Courcelle90} easily lift to our definitions.

A \emph{tree decomposition automaton} of width $\ell$ is a $5$-tuple $\autom = (\sigma^0, \sigma^1, \sigma^2, Q, F)$, where $Q$ is the set of \emph{states}, $F \subseteq Q$ is the set of \emph{accepting states}, and $\sigma^0$, $\sigma^1$, and $\sigma^2$ describe, informally speaking, how to compute the state of a node from the states of its children when it has, $0$, $1$, or $2$, children, respectively.
More formally,
\begin{itemize}
\item $\sigma^0$ is a function that takes as input a set $X$ of size at most $\ell+1$ and returns a state $\sigma^0(X) \in Q$,
\item $\sigma^1$ is a function that takes as input a state $q$ and two sets $X$, $Y$, each of size at most $\ell+1$, and returns a state $\sigma^1(q, X, Y) \in Q$, and
\item $\sigma^2$ is a function that takes as input two states $q_1$ and $q_2$, and three sets $X$, $Y$, $Z$, each of size at most $\ell+1$, and returns a state $\sigma^2(q_1, q_2, X, Y, Z)$.
\end{itemize}

The sets $X$, $Y$, and $Z$ taken by these functions are subsets of the vertex set $V(G)$ of the graph $G$ we are processing.
We assume (somewhat abusing the definitions), that the functions $\sigma^0$, $\sigma^1$, and $\sigma^2$ also receive the induced subgraphs $G[X]$, $G[Y]$, and $G[Z]$.
We also assume that the vertex set $V(G)$ is totally ordered, and the functions receive the orderings restricted to the sets $X$, $Y$, and $Z$.

On the algorithmic level, the functions $\sigma^0$, $\sigma^1$, and $\sigma^2$ are represented as Turing machines.
The \emph{evaluation time} of $\autom$, denoted by $\evaltime(\autom)$ is the maximum running time of these Turing machines.

The \emph{run} of $\autom$ on a binary tree decomposition $\Tc = (T,\bag)$ of width at most $\ell$ is the unique labeling $\run_{\autom} \colon V(T) \to Q$ so that
\begin{itemize}
\item if a node $t$ has no children, then $\run_{\autom}(t) = \sigma^0(\bag(t))$,
\item if a node $t$ has one child $c$, then $\run_{\autom}(t) = \sigma^1(\run_{\autom}(c), \bag(c), \bag(t))$, and
\item if a node $t$ has two children $c_1$, $c_2$, then $\run_{\autom}(t) = \sigma^2(\run_{\autom}(c_1), \run_{\autom}(c_2), \bag(c_1), \bag(c_2), \bag(t))$.
\end{itemize}

We say that $\autom$ accepts the pair $(G, \Tc)$ if $\run_{\autom}(r) \in F$ for the root $r$ of $T$.

Now we can state the well-known result that dynamic programming for $\CMSO_2$ can be expressed as a tree decomposition automaton.

\begin{lemma}[\cite{DBLP:journals/jal/ArnborgLS91,DBLP:journals/iandc/Courcelle90,DBLP:series/txcs/DowneyF13}]
\label{lem:cmsotreeautomaton}
There is an algorithm that, given a $\CMSO_2$-sentence $\Phi$ and an integer $\ell$, constructs a tree decomposition automaton $\autom = (\sigma^0, \sigma^1, \sigma^2, Q, F)$ of width $\ell$ so that 
\begin{itemize}
\item $\autom$ accepts $(G,\Tc)$ if and only if $G$ satisfies $\Phi$, and
\item $\max(|Q|, \evaltime(\autom)) \le f(\ell, |\Phi|)$ for a computable function $f$.
\end{itemize}
\end{lemma}

Again, we did not define the running time of the algorithm of \Cref{lem:cmsotreeautomaton}, but by definition it is bounded by a computable function on $|\Phi|$ and $\ell$.

\subsection{Graph-theoretic observations}
We then present some simple graph-theoretic observations that will be used multiple times in this paper.


We first recall the folklore results that graphs of treewidth at most $ w$ are sparse in many ways.

\begin{proposition}
\label{lem:twsparseness}
Every graph $G$ of treewidth at most $ w$
\begin{itemize}
\item has at most $|V(G)|\cdot w$ edges, and
\item contains an independent set of size at least $|V(G)|/(w+1)$.
\end{itemize}
\end{proposition}
\begin{proof}
The first result is Exercise~7.15 in \cite{DBLP:books/sp/CyganFKLMPPS15}. For the second result, let us remind that
 a graph $G$ is called $d$-degenerate if every subgraph of $G$ contains a vertex of degree
at most $d$. It is well-known, see e.g.  Exercise~7.14 in \cite{DBLP:books/sp/CyganFKLMPPS15}, that graphs of treewidth $w$ are $w$-degenerate. It is also 
  well-known and easy to prove that every $w$-degenerate graph $G$ can be properly colored in at most $w+1$ colors. Hence, $G$ contains an independent set of size at least $|V(G)|/(w+1)$.
  \end{proof}

%

We then observe that on graphs of small treewidth, one can always find a 
sufficiently large subset of a matching that is an induced matching.

\begin{lemma}
\label{lem:inducematching}
Let $G$ be a graph of treewidth $\tw(G) \le w$, and $M$ a matching in $G$.
There exists a subset $M' \subseteq M$ with $|M'| \ge |M|/(w+1)$ that is an induced matching in $G$.
\end{lemma}
\begin{proof}
Let $G_M$ be the graph whose vertex set is $M$, and two vertices of $G_M$ are adjacent if there is an edge between the vertices of the corresponding edges in $G$.
We observe that $G_M$ is a minor of $G$ (obtained by contracting $M$ and deleting all other vertices), and therefore has treewidth at most $w$.
By \Cref{lem:twsparseness}, $G_M$ has an independent set of size at least $|V(G_M)|/(w+1) = |M|/(w+1)$, and therefore $M$ contains a subset of size $\ge |M|/(w+1)$ that is an induced matching.
\end{proof}

\section{Bag signatures}\label{sec:bigsignatures}
The first ingredient of our algorithm is to enumerate for each bag of the given tree decomposition a family of subsets of the bag, so that the the intersection of ``the optimal solution'' with the bag is in this family.

For this, let us start by defining the optimal solution in a unique way.
Let a 4-tuple  $(G, \we, \Phi, w)$ be an instance of \mwisbt, and assume (without loss of generality) that there is a total order $\le_V$ on the set $V(G)$.
For two sets $X, Y \subseteq V(G)$, we say that $X$ is lexicographically larger than $Y$ if either (1) $|X| > |Y|$, or (2) $|X| = |Y|$ and when both are sorted, in the smallest position where they differ the element of $X$ is larger than the element of $Y$.
The two properties of this definition that we use are that (1) it is a total order on the subsets of $V(G)$, and (2) if $X' = (X \setminus \{u\}) \cup \{v\}$ for $u \in X$ and $v \notin X$, and $u <_V v$, then $X'$ is lexicographically larger than $X$.

A set $X \subseteq V(G)$ is a \emph{solution} to $(G, \we, \Phi, w)$ if $\tw(G[X]) \le w$ and $G[X]$ satisfies $\Phi$.
A solution $X$ is \emph{the optimal solution} if there is no other solution $X'$ so that either (1) $\we(X') > \we(X)$, or (2) $\we(X') = \we(X)$ and $X'$ is lexicographically larger than $X$.
Note that the optimal solution is indeed unique (if any solution exists).

This section is dedicated to the proof of the following lemma.

\begin{restatable}{lemma}{lembagsignatures}
\label{lem:bagsignatures}
There is an algorithm that, given an instance $(G, \we, \Phi, w)$ of \mwisbt, and a tree decomposition $\Tc = (T,\bag)$ of $G$ with $\mu(\Tc) \le k$, in time $f(k, w, |\Phi|) \cdot n^{\OO(kw)} \cdot |V(T)|$, for some computable function  $f(k, w, |\Phi|)$, computes for each $t \in V(T)$ a family $\fami(t)$ of subsets of $\bag(t)$, so that for all $t \in V(T)$,
\begin{itemize}
\item if $X \subseteq V(G)$ is the optimal solution to $(G, \we, \Phi, w)$, then $X \cap \bag(t) \in \fami(t)$, and
\item $|\fami(t)| \le f(k, w, |\Phi|) \cdot n^{\OO(kw)}$.
\end{itemize}
\end{restatable}

\subsection{Basic signatures}
\label{subsec:basicsigns}
We start the proof of \Cref{lem:bagsignatures} by establishing the properties of intersections of a  vertex set  $B$ with $\mu(B) \le k$ (view it as a bag $\bag(t)$ of a tree decomposition) and a set $X$ inducing a graph of small treewidth (view it as the optimal solution).

The first result is due to Yolov \cite{DBLP:conf/soda/Yolov18}.
\begin{lemma}[\cite{DBLP:conf/soda/Yolov18}]
\label{lem:isfam}
Let $G$ be an $n$-vertex graph and $B \subseteq V(G)$ a set with $\mu(B) \le k$.
Let $\isfam(B)$ be the set containing the intersection $I \cap B$ for every maximal independent set $I$ of $G$.
It holds that $|\isfam(B)| \le n^{\OO(k)}$, and $\isfam(B)$ can be computed in $n^{\OO(k)}$ time, when $G$, $B$, and $k$ are given.
\end{lemma}

Notice that $\isfam(B)$ contains all maximal independent sets of $G[B]$, but contains also non-maximal independent sets of $G[B]$.

We then establish that the set of edges of $G[X]$ incident to $B$ have a small vertex cover.

\begin{lemma}
\label{lem:vertexcover}
Let $G$ be a graph, $B \subseteq V(G)$ s.t. $\mu(B) \le k$, and $X \subseteq V(G)$ s.t. $\tw(G[X]) \le w$.
There exists a set $C \subseteq X$ of size $|C| =\OO(kw)$, so that every edge of $G[X]$ incident to $B$ is incident to $C$.
\end{lemma}
\begin{proof}
Let $G'$ be the subgraph of $G[X]$ consisting of the edges incident to $B$.
We claim that $G'$ has a vertex cover of size at most $2 k (w+1)$.
If the minimum vertex cover $C$ of $G'$ is larger than $2 k (w+1)$, then $G'$ contains a matching $M$ of size $|M| > k (w+1)$.
As $G'$ is a subgraph of $G$, $M$ is also a matching in $G$.
By \Cref{lem:inducematching}, $M$ has a subset $M'$ of size $|M'| > k$ that is an induced matching.
This contradicts that $\mu(B) \le k$. Hence all edges of  $G[X]$ incident to $B$ can be covered by at most  $2 k (w+1)$ vertices. 
\end{proof}

The next lemma shows that every independent set $S$ in  $B \setminus X$ that is dominated by $X$, can also be dominated by a small amount, $\OO(kw)$, vertices of $X$.
It will be applied to prune out unwanted vertices from the set $B \cap X$.


\begin{lemma}
\label{lem:throwaway}
Let $G$ be a graph, $B \subseteq V(G)$ a set with $\mu(B) \le k$, $X \subseteq V(G)$ a set with $\tw(G[X]) \le w$, and $S \subseteq B \setminus X$ an independent set in $G$.
There exists a set $D \subseteq X$ of size $|D| \le \OO(kw)$, so that every vertex of $S$ that is adjacent to $X$ is adjacent to $D$.
\end{lemma}
\begin{proof}
Let $S' = S \cap N(X)$, and let $D \subseteq X$ be an inclusion-wise minimal subset of $X$ so that $S' \subseteq N(D)$.
We claim that $|D| \le k (w+1)$.

Suppose otherwise.
Because $D$ is inclusion-wise minimal, each $v \in D$ has a ``private neighbor'' $s_v \in S' \cap N(v)$, meaning that $v$ is the only vertex in $D$ that is adjacent to $s_v$.
Because $D \subseteq X$ and $\tw(G[X]) \le w$, there exists a subset $D' \subseteq D$ so that $D'$ is an independent set and $|D'| \ge |D|/(w+1) > k$.
We claim that the matching $M$ obtained by matching each $v \in D'$ to $s_v$ is an induced matching.
Let $v s_v, u s_u \in M$.
Because $D'$ is an independent set, there is no edge between $u$ and $v$.
Because $S'$ is an independent set, there is no edge between $s_u$ and $s_v$.
Furthermore, because $s_u$ is a private neighbor of $u$, there is no edge between $v$ and $s_u$, and because $s_v$ is a private neighbor of $v$, there is no edge between $u$ and $s_v$.
Therefore, $M$ is an induced matching of size $>k$, whose every edge is incident to $S' \subseteq B$, which contradicts that $\mu(B) \le k$.
\end{proof}

Next we introduce the first part of our definition of a ``signature'' of a set $X \subseteq V(G)$ relative to a set $B \subseteq V(G)$ with $\mu(B) \le k$.
Let us remind that we use $\isfam(B)$ for  the family of  sets $I \cap B$  for every maximal independent set $I$ of $G$.

\begin{definition}[\textbf{Basic $B$-signature}]\label{def:b-signature}
Let $G$ be a graph $B \subseteq V(G)$ a set with $\mu(B) \le k$, and $X \subseteq V(G)$.
A \emph{basic $B$-signature} of $X$ is a triple $(C, S, D)$, such that
\begin{enumerate}
\item $C \subseteq X$ and every edge of $G[X]$ incident to $B$ is also incident to $C$,\label{def:basicsignature:item1}
\item $S \in \isfam(B)$ is a set such that $(X \cap B) \setminus C \subseteq S$, and
\item $D \subseteq X \setminus C$ is a set such that $(S \setminus C) \cap N(X \setminus C) \subseteq N(D)$.\label{def:basicsignature:item3}
\end{enumerate}
\end{definition}

Thus the set $C$ is a vertex cover of the edges incident to the set $B\cap X$. For every vertex of $X$ that is in  $S \setminus C$ all its neighbors in $G[X]$ must be in $C$ as otherwise $C$ would not be a vertex cover of the edges incident to $X \cap B$. Thus $C$ could be used to ``control'' such vertices. 
The vertex set $(X \cap B) \setminus C$ is an independent set and set $S$ is an independent set in $B$ containing $(X \cap B) \setminus C$.
The intuition of the set $D$ in \Cref{def:basicsignature:item3} is that  
 those vertices that are in $S \setminus C$ but have neighbors in $X \setminus C$ should be ``kicked out'' from $X$ by specifying this small set $D$.

By \Cref{lem:isfam}, all sets $S\in \isfam(B)$ can be enumerated in time  $n^{\OO(k)}$. The following lemma, which is built on  \Cref{lem:vertexcover,lem:throwaway}, yields that basic $B$-signatures for all potential solutions of treewidth at most $w$ could be enumerated  in time $n^{\OO(wk)}$.

\begin{lemma}
\label{lem:basicbsignexists}
For every set $B \subseteq V(G)$ with $\mu(B) \le k$ and $X \subseteq V(G)$ with $\tw(G[X]) \le w$, the set $X$ has a basic $B$-signature $(C, S, D)$ with $|C|,|D| \le  \OO(kw)$.
\end{lemma}
\begin{proof}
We first pick $C$ by applying \Cref{lem:vertexcover}.
Then, because $(X \cap B) \setminus C$ is an independent set, there exists a maximal independent set $S$ of $G[B]$ such that $(X \cap B) \setminus C \subseteq S$.
Moreover, $S \in \isfam(B)$, so we can pick it.
Finally, we obtain the set $D$ by applying \Cref{lem:throwaway} with the sets $B$, $S \setminus C$, and $X \setminus C$.
\end{proof}

Now the basic $B$-signature tells already a lot about the set $X \cap B$.
In particular, since $S\supseteq (X\cap B)\setminus C$, we have that
\begin{itemize}
\item[(i)] $C \cap B \subseteq B \cap X \subseteq S \cup (C \cap B)$.
\end{itemize}

This leaves the vertices in $B \cap (S \setminus C)$ undetermined by $S$ and $C$.
However, the set $D$ tells more about them.
Because $D \subseteq X \setminus C$, but by \Cref{def:basicsignature:item1} every edge of $G[X]$ incident to $B$ is also incident to $C$, we have that 
\begin{itemize}
\item[(ii)] no vertex in $(S \setminus C) \cap N(D)$ is in $X$.
\end{itemize}

However, basic $B$-signatures do not provide the following information about vertices of $S \setminus (N(D) \cup C)$:
\begin{itemize}
\item[(iii)] Vertices in $S \setminus (N(D) \cup C)$ may or may not be in $X$.
\end{itemize}

The vertices of type (iii) we call \emph{dangling}.
\begin{definition}[\textbf{Dangling vertex}]
For a basic $B$-signature $(C,S,D)$, vertex $v \in S \setminus (N(D) \cup C)$,  is a   \emph{dangling vertex} of this signature.
\end{definition}


\begin{lemma}
\label{lem:dangnbs}
For each dangling vertex $v$, it holds that $N(v) \cap X \subseteq C$.
\end{lemma}
\begin{proof}
Because $v \in S \setminus C$, but $v \notin N(D)$, $v$ is not adjacent to $X \setminus C$ by \Cref{def:basicsignature:item3} of \Cref{def:b-signature}. Therefore, all neighbors of $v$ that are in $X$ must also  be in $C$.
\end{proof}

\subsection{Understanding dangling vertices}
To encode $B \cap X$ completely in the signature, it remains to encode dangling vertices.
We will show that the dangling vertices intersect the optimal solution $X$ in a specific way.
For the following definitions, let us treat $X$ as the optimal solution, $B$ as a set with $\mu(B) \le k$, and $(C,S,D)$ as a basic $B$-signature of $X$.

We introduce some definitions to classify the dangling vertices.
We denote $\ell = \max(|C|, w)$.
The \emph{neighborhood tree decomposition} $\Tc_v$ of a dangling vertex $v$ is the rooted tree decomposition of the graph $G[(N(v) \cap C) \cup \{v\}]$ that has two bags, with the root bag containing $N(v) \cap C$ and the leaf bag containing $(N(v) \cap C) \cup \{v\}$.
The width of $\Tc_v$ is at most $|C| \le \ell$.

Let us then fix a tree decomposition automaton $\autom = (\sigma^0, \sigma^1, \sigma^2, Q, F)$ of width $\ell$, so that $\autom$ accepts a pair $(G, \Tc)$, where $\Tc$ is a tree decomposition of width $\le \ell$, if and only if $G$ satisfies $\Phi$, and furthermore, so that $|Q| \le f(\ell, |\Phi|)$ for a computable function $f$.
Such an automaton exists by \Cref{lem:cmsotreeautomaton}, and here we fix an arbitrary one.

Now, we define that the $\autom$-state of a dangling vertex $v$ is the state $q \in Q$ of $\autom$ on the neighborhood tree decomposition $\Tc_v$ of $v$.


Then, the \emph{type} of a dangling vertex $v$ is the triple $(\sign_v, N(v) \cap C, q_v)$, where $\sign_v \in \{<, \ge\}$ depending on whether $\we(v) < 0$ or $\we(v) \ge 0$, and $q_v \in Q$ is the $\autom$-state of $v$.
There are at most $2 \cdot 2^{|C|} \cdot |Q|$ different types of dangling vertices.

Let us prove our first structure lemma about dangling vertices.
It states that from each type, the set of dangling vertices in the optimal solution must form a suffix of the ordering of the dangling vertices by $\we$ and the vertex-total-order $\le_V$.

\begin{lemma}
\label{lem:dangvertstruct1}
Let $X$ be the optimal solution for $(G, \we, \Phi, w)$, $B \subseteq V(G)$ a set such that  $\mu(B) \le k$, and $(C,S,D)$ be a basic $B$-signature of $X$.
Let $U$ be a set of dangling vertices of $(C,S,D)$ of the same type. Then for every $v \in U \cap X$ and $u \in U \setminus X$, it holds that either (1) $\we(v) > \we(u)$, or (2) $\we(v) = \we(u)$ and $u \le_V v$.
\end{lemma}
\begin{proof}
Suppose that there is $v \in U \cap X$ and $u \in U \setminus X$ so that either (1) $\we(v) < \we(u)$ or (2) $\we(v) = \we(u)$ and $v \le_V u$.
Because $u$ and $v$ are of the same type, we have that $N(v) \cap C = N(u) \cap C$.  By \Cref{lem:dangnbs}, we have that $N(v) \cap X = N(u) \cap X$, and therefore the graphs $G[X]$ and $G[X \cup \{u\} \setminus \{v\}]$ are isomorphic.
It follows that $X \cup \{u\} \setminus \{v\}$ is a solution, and we observe that it would contradict that $X$ is the optimal solution.
\end{proof}

Now we know that for each type of dangling vertex, a suffix of the vertex ordering, according to both $\we$ and $\le_v$, must be selected into the set~$X$.
From this observation, one could already prove \Cref{lem:bagsignatures} with an additional factor of $n^{2^{\OO(kw)}}$.
However, in what follows, we optimize the encoding of the dangling vertices to eliminate any dependence on~$n$.

We now state our second structural lemma concerning dangling vertices.
It asserts that, for each type, one must select either only a few or almost all of the dangling vertices of that type.

\begin{restatable}{lemma}{lemdangvertstructtwo}
\label{lem:dangvertstruct2}
Let $X$ be the optimal solution for $(G, \we, \Phi, w)$, $B \subseteq V(G)$ a set with $\mu(B) \le k$, and $(C,S,D)$ a basic $B$-signature of $X$.
If $U$ is a set of dangling vertices of $(C,S,D)$ of the same type,  then either (1) $|U \cap X| \le \max(w+1, |Q|)$, or (2) $|U \setminus X| \le |Q|$.
\end{restatable}

We prove \Cref{lem:dangvertstruct2} through a sequence of auxiliary lemmas.
The key idea is to exploit the “finite-stateness” of the automaton~$\autom$.
We begin with a couple of lemmas about tree decompositions.

\begin{lemma}
\label{lem:tdsmallnbdang}
If $|U \cap X| > w$, then $|N(U) \cap X| \le w$.
\end{lemma}
\begin{proof}
If $|U \cap X| > w$ and $|N(U) \cap X| > w$, then $G[N[U] \cap X]$ would contain a complete bipartite graph with each side having more than $w$ as a subgraph.
However, such a complete bipartite graph has treewidth more than $w$, while the treewidth of $G[X]$ is at most  $ w$.
\end{proof}

\begin{lemma}
\label{lem:tdcompbipartite}
If $|U \cap X| > w+1$, then $G[X \setminus U]$ has a binary tree decomposition $\Tc = (T,\bag)$ of width at most $w$ that contains a leaf node $t$ with $\bag(t) = N(U) \cap X$.

\end{lemma}
\begin{proof}
Let $\Tc = (T,\bag)$ be an arbitrary binary tree decomposition of $G[X]$ of width at most $w$.

\begin{claim}
There is a node $t \in V(T)$ such that $N(U) \cap X \subseteq \bag(t)$.
\end{claim}
\begin{claimproof}
Recall that by \Cref{lem:dangnbs}, $N(v) \cap X = N(U) \cap X$ for all $v \in U$.
Let $\Tc = (T,\bag)$ be a tree decomposition of $G[X]$ of width at most $w$.
Suppose that there are $x_1, x_2 \in N(U) \cap X$ such that $\Tc$ does not have a bag containing both $x_1$ and $x_2$.
Let $R_1 \subseteq V(T)$ be the set of nodes $t$ of $T$ with $x_1 \in \bag(t)$ and $R_2 \subseteq V(T)$ the set of nodes $t$ with $x_2 \in \bag(t)$.
The sets $R_1$ and $R_2$ are disjoint.
However, for each $v \in U \cap X$, the set $R_v$ of nodes $t$ of $T$ with $v \in \bag(t)$ must overlap with both $R_1$ and $R_2$.
Because each of these sets are connected, it follows that there exists a bag of $\Tc$ containing all of $U \cap X$, but this is impossible since $|U \cap X| > w+1$.
It follows that for each pair $x_1, x_2 \in N(U) \cap X$ there is a bag containing both $x_1$ and $x_2$.
Because of the Helly property of connected subtrees, this implies that there is a bag containing $N(U) \cap X$.
\end{claimproof}

Now, by duplicating the node $t$ we can create a node $t'$ with $\bag(t') = \bag(t)$ with at most one child, and then we can add a new child to it with the bag of the child being exactly $N(U) \cap X$.
Finally, we delete all vertices of $U$ from the bags.
\end{proof}

\begin{lemma}
\label{lem:tdsetY}
Let $|U \cap X| > w$.  Then for any non-empty subset $Y \subseteq U$, enumerated as $Y = \{v_1, v_2, \ldots, v_\ell\}$ the graph $G[Y \cup (N(U) \cap X)]$ has a binary tree decomposition $\Tc = (T,\bag)$ of width at most $w$ such that
\begin{enumerate}
\item $T$ consists of a path $t_1, t_2, \ldots, t_{\ell}, t_{\ell+1}$ from the root $r$ to a leaf $t_{\ell+1}$ and for each $i \in [\ell]$ a two-edge path $t_i, a_i, b_i$, where $b_i$ is a leaf,
\item $\bag(t) = \bag(t_i) = N(U) \cap X$ for all $i \in [\ell+1]$,
\item $\bag(a_i) = N(U) \cap X$ for all $i \in [\ell]$, and
\item $\bag(b_i) = N[v_i] \cap X = \{v_i\} \cup (N(U) \cap X)$ for all $i \in [\ell]$.
\end{enumerate}
\end{lemma}
\begin{proof}
By \Cref{lem:tdsmallnbdang}, we have that $|N(U) \cap X| \le w$.
For each $v \in U$,  $N(v) \cap X = N(U) \cap X$. Therefore,  because $U$ is an independent set, the tree decomposition described in the statement is indeed a tree decomposition of $G[Y \cup (N(U) \cap X)]$. 
\end{proof}

Suppose now that $|U \cap X| > w$, so the tree decomposition described in \Cref{lem:tdsetY} exists for all subsets $Y \subseteq U$.
For a subset $Y \subseteq U$, let $\Tc_Y = (T_Y, \bag_Y)$ be the tree decomposition of $G[Y \cup (N(U) \cap X)]$ given by \Cref{lem:tdsetY}, where the ordering $\{v_1, \ldots, v_{\ell}\}$ of $Y$ is obtained from the total order $\le_V$.
Also, let $q_Y \in Q$ be the state of $\autom$ on $\Tc_Y$.
We then use the ``finite-stateness'' of $\autom$.

\begin{lemma}
\label{lem:dangvertaddordele}
If $|Y| \le |U|-|Q|$, then there exists $Y' \supsetneq Y$ with $q_{Y'} = q_{Y}$.
Also, if $|Y| \ge |Q|$, then there exists $Y' \subsetneq Y$ with $q_{Y'} = q_{Y}$.
\end{lemma}
\begin{proof}
Let us start by making observations on how $q_Y$ is computed.
Consider a run $\run_{\autom} \colon V(T_Y) \to Q$ of $\autom$ on $\Tc_Y$.
Each subtree of two nodes rooted at a vertex $a_i$ is isomorphic to the neighborhood tree decomposition $\Tc_v$ for the vertex $v \in U$ located in that subtree.
Therefore, if the type of the dangling vertices in $U$ is $(\sign, N, q)$, then $\run_{\autom}(a_i) = q$ for all $i \in [\ell]$.
Therefore, \[\run_{\autom}(t_i) = \transit^2(q, \run_{\autom}(t_{i+1}), N(U) \cap X, N(U) \cap X, N(U) \cap X).\]

It follows that $q_Y$ depends only on $|Y|$, so let us denote, for $m \ge 0$, the value $q_Y$ with $|Y| = m$ by $q_m$.
Furthermore, as the process of applying the function $\transit^2$ iteratively like above can be seen as a directed graph with vertex set $Q$ and each having only one out-edge, there must be integers $1 \le c_1, c_2 \le |Q|$ so that for all $m \ge c_1$, it holds that $q_m = q_{m-c_2}$.

It follows that if $|Y| \le |U|-|Q|$, we can add an arbitrary set of $c_2$ vertices from $U \setminus Y$ to $Y$ to create a strict superset $Y'$ with $q_{Y'} = q_{Y}$, and if $|Y| \ge |Q|$, we can remove an arbitrary set of $c_2$ vertices from $Y$ to create a strict subset $Y'$ with $q_{Y'} = q_{Y}$.
\end{proof}

We can now finish the proof of \Cref{lem:dangvertstruct2}, which we re-state here for convenience.

\lemdangvertstructtwo*
\begin{proof}
For the sake of contradiction, suppose that $|U \cap X| > \max(w+1, |Q|)$ and $|U \setminus X| > |Q|$.
Let $\Tc_1 = (T_1,\bag_1)$ be a binary tree decomposition of $G[X \setminus U]$ of width at most $w$ that has a leaf node $t$ with $\bag(t) = N(U) \cap X$, which exists by \Cref{lem:tdcompbipartite}.
Let also $\Tc_{U \cap X}$ be the tree decomposition of $G[(U \cap X) \cup (N(U) \cap X)]$ given by \Cref{lem:tdsetY}, and denote by $q_{U \cap X}$ the state of $\autom$ on $\Tc_{U \cap X}$.
We observe that by adding $\Tc_{U \cap X}$ to $\Tc_1$ as a child of $t$, we obtain a tree decomposition $\Tc$ of $G[X]$.
Moreover, the state of $\autom$ on $\Tc$ depends only on $\Tc_1$ and $q_{U \cap X}$.
We then consider two cases.

First, suppose that the signature of each vertex in $U$ is $(\ge, N(U) \cap C, q)$.
By \Cref{lem:dangvertaddordele}, there exists a set $Y$ with $U \cap X \subsetneq Y \subseteq U$ so that $q_{Y} = q_{U \cap X}$.
By adding $\Tc_Y$ to $\Tc_1$ as a child of $t$ we obtain a tree decomposition $\Tc'$ of $G[X \cup Y]$, and the state of $\autom$ on $\Tc'$ is the same as the state of $\autom$ on $\Tc$.
Therefore, $X \cup Y$ is a solution, but $\we(X \cup Y) \ge \we(X)$ and $|X \cup Y| > |X|$, so this contradicts that $X$ is the optimal solution.

Then, suppose that the signature of each vertex in $U$ is $(<, N(U) \cap C, q)$.
By \Cref{lem:dangvertaddordele}, there exists a set $Y \subsetneq U \cap X$ so that $q_{Y} = q_{U \cap X}$.
By adding $\Tc_Y$ to $\Tc_1$ as a child of $t$ we obtain a tree decomposition $\Tc'$ of $G[X \cup Y]$, and the state of $\autom$ on $\Tc'$ is the same as the state of $\autom$ on $\Tc$.
Therefore, $(X \setminus U) \cup Y$ is a  solution. However, since $\we((X \setminus U) \cup Y) > \we(X)$, this contradicts the assumption that $X$ is the optimal solution.
\end{proof}

\subsection{Proof of \Cref{lem:bagsignatures}}
We then put together the results of \Cref{subsec:basicsigns} and \Cref{lem:dangvertstruct1,lem:dangvertstruct2} to finish the proof of \Cref{lem:bagsignatures}, which we re-state here for convenience.

\lembagsignatures*
\begin{proof}
It suffices to give an algorithm that, given a set $B \subseteq V(G)$ with $\mu(B) \le k$, computes a family $\fami$ of subsets of $B$, so that for the optimal solution $X$ (if one exists), it is guaranteed that $X \cap B \in \fami$.
This algorithm can then be applied for each bag of $\Tc$, causing an overhead of a factor of $|V(T)|$ in the running time.

We first describe the algorithm and then show its correctness.

We first apply \Cref{lem:isfam} to compute a set $\isfam(B)$ of independent subsets of $B$, so that for each maximal independent set $I$ of $G$, $I \cap B \in \isfam(B)$.
We have that $|\isfam(B)| \le n^{\OO(k)}$ and the running time of the enumeration algorithm is $n^{\OO(k)}$.
We iterate over all $I_B \in \isfam(B)$ and all pairs of sets $C,D \subseteq V(G)$ with $|C|,|D| \le \OO(kw)$ (where the constant in the $\OO$-notation comes from \Cref{lem:basicbsignexists}), and consider the basic $B$-signature $(C,S,D)$.

We fix a tree decomposition automaton $\autom = (\transit^0, \transit^1, \transit^2, Q, F)$ of width $\ell = \max(|C|, w)$, so that $\autom$ accepts a pair $(G,\Tc)$ if and only if $G$ satisfies $\Phi$, and $|Q|, \evaltime(\autom) \le f(\ell, |\Phi|)$ for a computable function.
We compute the set of dangling vertices $U = S \setminus (N(D) \cup C)$, and for each $v \in U$, we compute its type  (with respect to $\autom$).
We partition $U$ into at most $2 \cdot 2^{|C|} \cdot |Q|$ many sets $U_1, \ldots, U_m$ based on their types, and sort each set $U_i$ based on the weight-function $\we$ and the vertex-order $\le_V$.
For each $U_i$, let $\mathcal{U}_i$ be the family of subsets of $U_i$ that consists of the suffixes of $U_i$ of length at most $\max(|Q|, w+1)$ or at least $|U_i| - |Q|$.
We have that $|\mathcal{U}_i| \le 2 \cdot |Q|$.

There are at most   $ (2 \cdot |Q|)^{m} \le (2 \cdot |Q|)^{2 \cdot 2^{|C|} \cdot |Q|}$ ways to select  $U'_i \in \mathcal{U}_i$ for each $i \le m$; we iterate over all selections. 
For each selection, we add set $S \cup (C \cap B) \cup \bigcup_i U'_i$ to the family $\fami$.


The constructed family $\fami$  has at most $g(k, w, |\Phi|) \cdot n^{\OO(kw)}$ sets in it, and we can construct it in time $g(k, w, |\Phi|) \cdot n^{\OO(kw)}$, where $g$ is a computable function.
It remains to prove that $X \cap B \in \fami$.

By \Cref{lem:basicbsignexists}, there exists a basic $B$-signature $(C,S,D)$ of $X$ so that $|C|,|D|= \OO(kw)$ and $S \in \isfam(B)$.  In particular, we have  that $B \cap X \subseteq S \cup (C \cap B)$ and no vertex from  $B \cap N(D) \setminus C$ is in $X$.
We claim that the iteration that considers the $B$-signature $(C,S,D)$ adds $B \cap X$ to $\fami$.
Let $U = S \setminus (N(D) \cup C)$ be the dangling vertices, and let $U_1, \ldots, U_m$ be their partition based on their types.
By \Cref{lem:dangvertstruct1,lem:dangvertstruct2}, when each $U_i$ is sorted according to $\we$ and the vertex-order $\le_V$, the set $X \cap U_i$ consists of either a suffix of length at most $ \max(w+1, |Q|)$, or a suffix of length at least $ |U_i| - |Q|$.
Therefore, $X \cap U_i \in \mathcal{U}_i$, and thus $\fami$ contains $X \cap B$.
\end{proof}

\section{Inner tree decomposition}
\label{sec:innertree}
In this section, we show that if $\Tc$ is a tree decomposition of $G$ with $\mu(\Tc) \le k$, and $X \subseteq V(G)$ with $\tw(G[X]) \le w$, then we can construct a tree decomposition of $G[X]$ of width $\OO(kw^2)$, that ``follows'' the structure of $\Tc$ in a certain precisely defined way.
This is the main ingredient of our dynamic programming algorithm.

First, we will make sure that the tree decomposition $\Tc$ of $G$ itself is very nicely structured.
For this, we introduce several definitions of \emph{nice tree decomposition} (see e.g.~\cite{DBLP:books/sp/Kloks94,DBLP:books/sp/CyganFKLMPPS15,fomin2019kernelization} for nice tree decompositions in the literature).
We define that a rooted tree decomposition $\Tc = (T,\bag)$ is \emph{nice} if each node $t \in V(T)$ either
\begin{enumerate}
\item is a leaf and has $\bag(t) = \emptyset$ (initial-node),
\item has one child $c$, which has $\bag(c) = \bag(t) \setminus v$ for $v \in \bag(t)$ ($v$-introduce-node),
\item has one child $c$, which has $\bag(c) = (\bag(t) \cup \{v\})$ for $v \notin \bag(t)$ ($v$-forget-node),
\item has two children $c_1$, $c_2$, which satisfy $\bag(t) = \bag(c_1) = \bag(c_2)$ (join-node), or 
\item has one child $c$, which has $\bag(c) = \bag(t)$ (neutral-node).
\end{enumerate}

Furthermore, we say that a rooted tree decomposition $\Tc$ is \emph{supernice} if it is nice, and additionally
\begin{enumerate}[resume]
\item for each $v \in V(G)$, the node $\topnode_{\Tc}(v)$ is a neutral-node ($v$-top-node), and
\item for the root $r$ it holds that $\bag(r) = \emptyset$.
\end{enumerate}
In the context of supernice tree decompositions, top-nodes will be called only top-nodes, not neutral-nodes.

Furthermore, we say that $\Tc$ is $\ell$-supernice, if it is supernice, and for each node $t$ that is either an introduce-node, a forget-node, or a join-node, there is a path consisting of $\ell+1$ nodes, starting from the parent of $t$ and ending at its ancestor, such that each node on that path is a neutral-node.
We then observe that by applying standard ideas, any tree decomposition can be turned into $\ell$-supernice.

\begin{lemma}
\label{lem:tdintosupernice}
There is an algorithm that, given a tree decomposition $(T,\bag)$ of a graph $G$ and an integer $\ell$, in time $(|V(T)|+|V(G)|+\ell)^{\OO(1)}$ returns an $\ell$-supernice tree decomposition $(T',\bag')$ of $G$, such that 
\begin{itemize}
\item for each $t' \in V(T')$, there exists $t \in V(T)$ with $\bag'(t') \subseteq \bag(t)$,
\item for each $t \in V(T)$, there exists $t' \in V(T')$ with $\bag'(t') = \bag(t)$, and
\item $|V(T')| \le (|V(T)|+|V(G)|+\ell)^{\OO(1)}$.
\end{itemize}
\end{lemma}
\begin{proof}
We edit the given three decomposition in several phases.
First, we make sure that if $t$ is a node with degree $\ge 3$, then $\bag(t) = \bag(s)$ for all neighbors $s$ of $t$.
We do this by subdividing each edge of $T$ adjacent to $t$, and let the bags of the subdivision nodes be $\bag(t)$.
This increases $|V(T)|$ by a factor of at most $3$, and can be implemented in $\OO(|V(T)| \cdot |V(G)|)$ time.
Then, we make the maximum degree of $T$ be $3$ by replacing nodes of degree $>3$ by binary trees.
This increases $|V(T)|$ by a factor of at most $3$, and can be implemented in $\OO(|V(T)| \cdot |V(G)|)$ time.
Then, for each leaf $t$ of $T$, we add a new node with empty bag adjacent to $t$, such that in the resulting tree decomposition, every leaf has an empty bag.
Then, we consider each path $P$ of $T$, whose endpoints have degree $\neq 2$, and internal nodes have degree $2$.
If there are two consecutive nodes $s,t$ on such a path with $|\bag(s) \setminus \bag(t)| + |\bag(t) \setminus \bag(s)| \ge 2$, we insert nodes whose bags are subsets of $\bag(s)$ and $\bag(t)$ between them to ensure that $|\bag(s) \setminus \bag(t)| + |\bag(t) \setminus \bag(s)| \le 1$ for any two consecutive nodes $s,t$.
This increases $|V(T)|$ by a factor of $\OO(|V(G)|)$ and can be implemented in $\OO(|V(T)| \cdot |V(G)|^2)$ time.
After that, for every internal node $t$ of such a path, we replace the node by a path of length $\ell+3$ whose all bags are $\bag(t)$.
This increases $|V(T)|$ by a factor of $\OO(\ell)$, and can be implemented in time $\OO(|V(T)| \cdot |V(G)| \cdot \ell)$.

We root the obtained tree decomposition at an arbitrary leaf node, and observe that the resulting rooted tree decomposition is $\ell$-supernice.
Furthermore, it satisfies the other required properties.
\end{proof}

We can then state the main lemma of this section, which now formally states what it means for a tree decomposition of $G[X]$ to ``follow'' a tree decomposition $\Tc$ of $G$.

\begin{lemma}
\label{lem:virtdecomp}
For any two integers $k$ and $w$, there exists an integer $\ell \le \OO(k w^2)$ such that the following holds:
Let $G$ be a graph, $\Tc = (T,\bag)$ an $\ell$-supernice tree decomposition of $G$ with $\mu(\Tc) \le k$, and $X \subseteq V(G)$ such that $\tw(G[X]) \le w$.
There exists a bipartition $(X_1, X_2)$ of $X$ and a mapping $\vbag \colon V(T) \to 2^{V(G)}$ such that
\begin{enumerate}
\item $\hat{\Tc} = (T,\vbag)$ is a tree decomposition of $G[X]$ of width at most $\ell$,\label{lem:virtdecomp:prop1}
\item $\bag(t) \cap X_1 \subseteq \vbag(t) \cap X_1 \subseteq \subtree_{\Tc}(t) \cap X_1$ for all $t \in V(T)$,\label{lem:virtdecomp:prop2}
\item for $v \in X_2$, the only node $t \in V(T)$ with $v \in \vbag(t)$ is $t = \topnode_{\Tc}(v)$,\label{lem:virtdecomp:prop3}
\item and furthermore, each $t \in V(T)$ satisfies:\label{lem:virtdecomp:prop4}
\begin{enumerate}
\item if $t$ is a $v$-introduce-node with child $c$ in $\Tc$, then either $\vbag(t) = \vbag(c)$ or $\vbag(t) = \vbag(c) \cup \{v\}$,
\item if $t$ is a $v$-forget-node with child $c$ in $\Tc$, then either $\vbag(t) = \vbag(c)$ and $v \notin X_2$, or $\vbag(t) = \vbag(c) \setminus \{v\}$ and $v \in X_2$,
\item if $t$ is a join-node with children $c_1,c_2$ in $\Tc$, then $\vbag(t) = \vbag(c_1) \cup \vbag(c_2)$ and $\vbag(c_1) \cap \vbag(c_2) = \vbag(t) \cap \bag(t)$,
\item if $t$ is a neutral-node with child $c$ in $\Tc$, then either $\vbag(t) = \vbag(c)$, or $\vbag(t) = \vbag(c) \setminus \{v\}$ for some $v \in X_1$,
\item if $t$ is a $v$-top-node with child $c$ in $\Tc$, then either $\vbag(t) = \vbag(c)$ and $v \notin X_2$, or $\vbag(t) = \vbag(c) \cup \{v\}$ and $v \in X_2$, and
\item if $t$ is the root of $T$, then $\vbag(t) = \emptyset$.
\end{enumerate}
\end{enumerate}
\end{lemma}


This section is devoted primarily to the proof of \Cref{lem:virtdecomp}. Although we establish \Cref{lem:virtdecomp} via a specific construction, our later applications of the lemma do  not require that $(X_1,X_2)$ or $(T,\vbag)$ actually follow this construction; we need only that they satisfy the properties stated in \Cref{lem:virtdecomp:prop1,lem:virtdecomp:prop2,lem:virtdecomp:prop3,lem:virtdecomp:prop4}.
The main reason why the statement of \Cref{lem:virtdecomp} describes $\hat{\Tc}$ in such detail is because it allows to reduce the number of cases we must handle later when presenting the dynamic-programming procedure.

We observe that in addition to the explicit conditions listed in \Cref{lem:virtdecomp}, they imply the following.

\begin{lemma}
\label{lem:virtdecomp-furtherconstraints}
Any bipartition $(X_1,X_2)$ and a tree decomposition $\hat{\Tc} = (T, \vbag)$ satisfying the requirements of \Cref{lem:virtdecomp} also satisfies
\begin{enumerate}
\item $X_2$ is an independent set, and
\item for all nodes $x \in V(T)$, $N(X_2 \cap \bag(x)) \cap X_1 \cap \subtree_{\Tc}(x) \subseteq \vbag(x)$.\label{lem:virtdecomp-furtherconstraints:item2}
\end{enumerate}
\end{lemma}
\begin{proof}
Set $X_2$ is independent because each vertex $v \in X_2$ occurs only in the bag of the node $\top_{\Tc}(v)$, and these nodes are distinct for distinct $v$.

For \Cref{lem:virtdecomp-furtherconstraints:item2}, suppose this does not hold for a node $x \in V(T)$, and let $v$ be a vertex such that $v \in N(X_2 \cap \bag(x)) \cap X_1 \cap \subtree_{\Tc}(x)$ but $v \notin \vbag(x)$.
Furthermore, let $u \in X_2 \cap \bag(x)$ such that $uv \in E(G)$.
As $v \in \subtree_{\Tc}(x)$ and $u \in \bag(x)$, $x$ must be on a path from a node $y$ with $v \in \bag(y)$ to the node $z = \topnode_{\Tc}(u)$.
Because $z$ is the only node such that $u \in \vbag(z)$, we have that $v \in \vbag(z)$.
However, also $v \in \vbag(y)$, so the connectedness-condition of tree decompositions implies $v \in \vbag(x)$.
\end{proof}

\subsection*{Proof of \Cref{lem:virtdecomp}}
The rest of this section is dedicated to the proof of \Cref{lem:virtdecomp}.

We let $X_H \subseteq X$ be the vertices in $X$ that have degree $> 2 (w+1)$ in $G[X]$, and $X_L = X \setminus X_H$.
Then, we let $X_{LL} \subseteq X_L$ be the vertices in $X_L$ that have at least one neighbor in $X_L$.
We construct the bipartition $(X_1,X_2)$ by letting
\[(X_1,X_2) = (X_H \cup X_{LL}, X_L \setminus X_{LL}).\]

We then construct the tree decomposition $\hat{\Tc} = (T,\vbag)$.
We first construct a tree decomposition $(T,\vbag')$ that satisfies the properties of \Cref{lem:virtdecomp:prop1,lem:virtdecomp:prop2,lem:virtdecomp:prop3}, and then use the fact that $\Tc$ is $\ell$-supernice to modify $(T,\vbag')$ such that it satisfies \Cref{lem:virtdecomp:prop4}.

\paragraph{Constructing $(T, \vbag')$.}
Let us now construct $(T,\vbag')$.
Intuitively, the idea is to take $(T,\bag)$ restricted to $X$, then restrict the vertices in $X_2$ only to their top bags, and then extend vertices in $X_H$ upwards such that they reach the top bags of all of their neighbors in $X_2$.

For $t \in V(T)$, denote by $\mathsf{outH}(t)$ the vertices in $X_H \cap \subtree_{\Tc}(t)$ that have at least one neighbor in $\bag(t) \cap X_L$.
We also denote by $\topv_{\Tc}(t)$ the set of vertices $v$ such that $\topnode_{\Tc}(v) = t$.
Because $\Tc$ is nice, $|\topv_{\Tc}(t)| \le 1$.
We let
\[\vbag'(t) = (\bag(t) \cap X_1) \cup \mathsf{outH}(t) \cup (\topv_{\Tc}(t) \cap X_2).\]

Let us first prove that $(T,\vbag')$ is a tree decomposition of $G[X]$.

\begin{lemma}
\label{lem:virtdecomp-istd}
The pair $(T,\vbag')$ is a tree decomposition of $G[X]$.
\end{lemma}
\begin{proof}
The vertex-condition is satisfied because each $v \in X$ occurs at least in $\vbag'(\topnode_{\Tc}(v))$.

For edges between two vertices in $X_1$, the edge-condition follows from that $\vbag'(t) \cap X_1 \supseteq \bag(t) \cap X_1$.
As vertices in $X_2$ are adjacent in $X$ only to vertices in $X_H$, it remains to consider edges between $X_2$ and $X_H$.
Let $v \in X_2$ and $u \in X_H$ with $uv \in E(G)$.
Because every node whose bag contains $v$ is a descendant of $\topnode_{\Tc}(v)$, it must be that $u \in \subtree_{\Tc}(\topnode_{\Tc}(v))$.
Therefore, $u \in X_H \cap \subtree_{\Tc}(\topnode_{\Tc}(v))$, and $u$ has a neighbor in $\bag(\topnode_{\Tc}(v)) \cap X_L$, so $u \in \mathsf{outH}(\topnode_{\Tc}(v))$, so $\{u,v\} \subseteq \vbag'(\topnode_{\Tc}(v))$.

The connectedness condition is trivial for vertices in $X_{LL}$, as they occupy the same bags in both $(T,\bag)$ and $(T,\vbag')$.
It is also trivial for vertices in $X_2$, as they occur in exactly one bag.
It remains to consider vertices in $X_H$.
For a vertex $v \in X_H$, the set of nodes whose bags contain $v$ in $(T,\vbag')$ is a superset of the set of nodes whose bags contain $v$ in $\Tc$.
Therefore, it suffices to prove that if $v \in \vbag'(t) \setminus \bag(t)$, then $v \in \vbag'(c)$ for a child $c$ of $t$.

If $v \in \vbag'(t) \setminus \bag(t)$, then $v \in \mathsf{outH}(t)$, so $v \in \subtree_{\Tc}(t)$, so it follows that there is exactly one child $c$ of $t$ with $v \in \subtree_{\Tc}(c)$.
Let $u$ be the vertex in $\bag(t) \cap X_L$ that is adjacent to $v$.
Because of the connectedness condition of $\Tc$, it must be that $u \in \bag(c)$, and therefore $v \in \mathsf{outH}(c)$, so $v \in \vbag'(c)$.
\end{proof}

We then bound the width of $(T,\vbag')$.

\begin{lemma}
\label{lem:virtdecomp-width}
For each $t \in V(T)$, $|\vbag'(t)| \le \OO(k w^2)$.
\end{lemma}
\begin{proof}
The set $\vbag'(t)$ is the union of the four sets $\bag(t) \cap X_H$, $\bag(t) \cap X_{LL}$, $\mathsf{outH}(t)$, and $\topv_{\Tc}(t) \cap X_2$.
Obviously, $|\topv_{\Tc}(t) \cap X_2| \le 1$.
Let us bound the sizes of the three other sets one by one.

\begin{claim}
\label{lem:virtdecomp-width:claim1}
$|\bag(t) \cap X_H| \le k \cdot (w+1)^2$.
\end{claim}
\begin{claimproof}
Suppose otherwise.
As $\tw(G[\bag(t) \cap X_H]) \le w$, there is a subset $X_H' \subseteq \bag(t) \cap X_H$ that is an independent set and has size $|X_H'| > k \cdot (w+1)$.
Consider the bipartite graph $B$ formed by the edges between $X_H'$ and $N(X_H') \cap X$, and note that $\tw(B) \le \tw(G[X]) \le w$.
Because $X_H'$ is an independent set, and each of its vertices has degree $> 2 (w+1)$ in $G[X]$, each of its vertices has also degree $> 2 (w+1)$ in the bipartite graph.
Therefore, for any subset $Y \subseteq X_H'$, it must hold that $Y$ has at least $|Y|$ neighbors in $B$, as otherwise we would obtain a subgraph of $B$ with $< 2 \cdot |Y|$ vertices but $> 2 (w+1) \cdot |Y|$ edges, implying by \Cref{lem:twsparseness} that $B$ would have treewidth $>w$.
Now, Hall's Theorem implies that $B$ contains a matching $M$ of size $|M| = |X_H'| > k \cdot (w+1)$ saturating $X_H'$.
By \Cref{lem:inducematching}, then $G$ contains an induced matching of size $|M|/(w+1) > k$ whose all edges are incident to $\bag(t)$, contradicting that $\mu(\Tc) \le k$.
\end{claimproof}

\begin{claim}
\label{lem:virtdecomp-width:claim2}
$|\bag(t) \cap X_{LL}| \le k \cdot (w+1) \cdot (2 (w+1)+1)$.
\end{claim}
\begin{claimproof}
Suppose otherwise.
As each vertex in $X_{LL}$ has a neighbor in $X_L$, but each vertex in $X_L$ has degree $\le 2 (w+1)$ in $G[X]$, there is a matching $M$ of size $|M| \ge |\bag(t) \cap X_{LL}| / (2 (w+1) + 1) > k \cdot (w+1)$ in $G[X_L]$ whose all edges are incident to $\bag(t) \cap X_{LL}$.
By \Cref{lem:inducematching}, then $G$ contains an induced matching of size $|M|/(w+1) > k$ whose all edges are incident to $\bag(t) \cap X_{LL}$, contradicting that $\mu(\Tc) \le k$.
\end{claimproof}

\begin{claim}
\label{lem:virtdecomp-width:claim3}
$|\mathsf{outH}(t)| \le k \cdot (w+1) \cdot 2 (w+1)$.
\end{claim}
\begin{claimproof}
Suppose otherwise.
Because each vertex in $\mathsf{outH}(t)$ has a neighbor in $\bag(t) \cap X_L$ and each vertex in $\bag(t) \cap X_L$ has degree $\le 2 (w+1)$ in $G[X]$, there is a matching $M$ of size $|M| \ge |\mathsf{outH}(t)|/(2 (w+1)) > k \cdot (w+1)$ between $\mathsf{outH}(t)$ and $\bag(t) \cap X_L$.
By \Cref{lem:inducematching}, then $G$ contains an induced matching of size $|M|/(w+1) > k$ whose all edges are incident to $\bag(t) \cap X_L$, contradicting that $\mu(\Tc) \le k$.
\end{claimproof}
\Cref{lem:virtdecomp-width:claim1,lem:virtdecomp-width:claim2,lem:virtdecomp-width:claim3} complete the proof.
\end{proof}

\Cref{lem:virtdecomp-istd,lem:virtdecomp-width} establish that $(T,\vbag')$ satisfies \Cref{lem:virtdecomp:prop1}, while \Cref{lem:virtdecomp:prop2,lem:virtdecomp:prop3} easily follow from the construction of $(T,\vbag')$.

\paragraph{Modifying $(T,\vbag')$ into $(T,\vbag)$.}
We then modify $(T,\vbag')$ into $(T,\vbag)$.
The only reason why $(T,\vbag')$ does not already satisfy \Cref{lem:virtdecomp:prop4} is that many vertices may be forgotten at once, i.e., for a forget-node $t$ with a child $c$ it may hold that $|\bag(c) \setminus \bag(t)| \ge 2$.
The remedy for this will be that these vertices will instead be forgotten one-by-one in the $\ell+1$ neutral-nodes above $t$, which exists because $\Tc$ is $\ell$-supernice.

We then make the above discussion more formal.
Let us start by showing what happens at neutral-nodes and forget-nodes in $(T,\vbag')$.
(Recall that a top-node is not a neutral-node here.)

\begin{lemma}
\label{lem:innerdecompneutralnode}
If $t$ is a neutral-node with child $c$ in $\Tc$, then $\vbag'(t) = \vbag'(c)$.
\end{lemma}
\begin{proof}
The construction of $\vbag'(t)$ (after choosing $X$) depends only on $\bag(t)$ and $\subtree(t)$, which are the same for $t$ and $c$.
\end{proof}

\begin{lemma}
If $t$ is a $v$-forget-node with child $c$ in $\Tc$, then $\vbag'(t) \subseteq \vbag'(c)$.
\end{lemma}
\begin{proof}
Suppose first that $v \in X$.
As $\bag(t) \subseteq \bag(c)$ and $\topv_{\Tc}(t) = \emptyset$, we have that $(\bag(t) \cap X_1) \cup (\topv_{\Tc}(t) \cap X_1) \subseteq (\bag(c) \cap X_1) \cup (\topv_{\Tc}(c) \cap X_1)$.
Furthermore, because $\subtree_{\Tc}(t) = \subtree_{\Tc}(c)$ and $\bag(t) \subseteq \bag(c)$, it holds that $\mathsf{outH}(t) \subseteq \mathsf{outH}(c)$.
\end{proof}

We then obtain $(T,\vbag)$ from $(T,\vbag')$ as follows.
Let $t$ be a $v$-forget-node of $\Tc$ with child $c$, and let $F = \vbag'(c) \setminus \vbag'(t)$.
Let $t, a_1, a_2, \ldots, a_{\ell+1}$ be the path from $t$ to its ancestor $a_{\ell+1}$, such that each $a_1, \ldots, a_{\ell+1}$ is a neutral-node.
Let $F_1, \ldots, F_{\ell+1}$ be a sequence of subsets of $F \setminus \{v\}$, such that $|F_1| \le 1$, $F_i \subseteq F_{i+1}$, $|F_{i+1}| \le |F_i|+1$, and $F_{|F \setminus \{v\}|} = F \setminus \{v\}$.
Such a sequence exists because $|F| \le \ell+1$.
We observe that $v$ is the only vertex of $F$ that can be in $X_2$, the others must be in $X_1$.

If $v \in X_2$, which implies that $v \in \vbag'(c)$ but $v \notin \vbag'(t)$, we let $\vbag(t) = \vbag'(c) \setminus \{v\}$, and $\vbag(a_i) = \vbag'(c) \setminus (F_i \cup \{v\})$.
Note that $\vbag(a_{\ell+1}) = \vbag'(t) = \vbag'(a_{\ell+1})$.

Otherwise, if $v \notin X_2$, we let $\vbag(t) = \vbag'(c)$, $\vbag(a_i) = \vbag'(c) \setminus F_i$, and $\vbag(a_{\ell+1}) = \vbag'(t)$.
Note that in particular, if $v \in F$, then $\vbag(a_{\ell}) = \vbag'(t) \cup \{v\}$.

For the other nodes $t$, for which the construction of $\vbag(t)$ was not described above, we let $\vbag(t) = \vbag'(t)$.
We observe that $(T,\vbag)$ is a tree decomposition of $G[X]$, because it is obtained from $(T,\vbag')$ by just ``lenghtening'' the subtrees corresponding to some vertices upwards.
Also, the width of $(T,\vbag)$ is no larger than the width of $(T,\vbag')$, because each bag of $(T,\vbag)$ is a subset of a bag of $(T,\vbag')$.

Clearly, the construction of $(T,\vbag)$ satisfies the requirement of \Cref{lem:virtdecomp:prop4} for forget-nodes.
For unmodified neutral-nodes it satisfies it by \Cref{lem:innerdecompneutralnode}.
For the neutral-nodes modified by the above construction, we recall first that the only vertex of $F$ that can be in $X_2$ is $v$, so the construction forgets only vertices from $X_1$ in neutral-nodes.

It remains to prove that the construction satisfies the requirements for introduce-nodes, join-nodes, top-nodes, and for the root.

\begin{lemma}
If $t$ is a $v$-introduce-node with child $c$ in $\Tc$, then either $\vbag(t) = \vbag(c)$ or $\vbag(t) = \vbag(c) \cup \{v\}$.
\end{lemma}
\begin{proof}
We observe that $\vbag(t) = \vbag'(t)$ and $\vbag(c) = \vbag'(c)$, so it suffices to argue the conclusion for $\vbag'$.

As $\bag(t) = \bag(c) \cup \{v\}$, we have that $\bag(t) \cap X_1 \subseteq (\bag(c) \cap X_1) \cup \{v\}$.
Because $t$ is an introduce node, also $\topv_{\Tc}(t) = \topv_{\Tc}(c) = \emptyset$.
It remains to argue that $\mathsf{outH}(t) \subseteq \mathsf{outH}(c) \cup \{v\}$.
Suppose there is $u \in \mathsf{outH}(t) \setminus \{v\}$ that is not in $\mathsf{outH}(c)$.
There exists $w \in \bag(t) \cap X_L$ that is a neighbor of $u$.
Now, $w \in \bag(t) \cap X_L$, so $u \in \mathsf{outH}(c)$.
\end{proof}

\begin{lemma}
If $t$ is a join-node with children $c_1,c_2$ in $\Tc$, then $\vbag(t) = \vbag(c_1) \cup \vbag(c_2)$ and $\vbag(c_1) \cap \vbag(c_2) = \vbag(t) \cap \bag(t)$.
\end{lemma}
\begin{proof}
We observe that $\vbag(t) = \vbag(c_1) = \vbag(c_2) = \vbag'(t) = \vbag'(c_1) = \vbag'(c_2)$, so it suffices to argue the conclusion for $\vbag'$.

Because $\Tc$ is $\ell$-supernice, both $c_1$ and $c_2$ are neutral-nodes, so $\topv_{\Tc}(t) = \topv_{\Tc}(c_1) = \topv_{\Tc}(c_2) = \emptyset$.
Because $\bag(t) = \bag(c_1) = \bag(c_2)$, we also have that $\bag(t) \cap X_1 = \bag(c_1) \cap X_1 = \bag(c_2) \cap X_1$.
Furthermore, because $\bag(t) = \bag(c_1) = \bag(c_2)$ and $\subtree_{\Tc}(t) = \subtree_{\Tc}(c_1) \cup \subtree_{\Tc}(c_2)$, we have that $\mathsf{outH}(t) = \mathsf{outH}(c_1) \cup \mathsf{outH}(c_2)$.
It follows that $\vbag'(t) = \vbag'(c_1) \cup \vbag'(c_2)$.
To establish that $\vbag(c_1) \cap \vbag(c_2) = \vbag(t) \cap \bag(t)$, we observe that if $v \in \subtree_{\Tc}(c_1) \cap \subtree_{\Tc}(c_2)$, then $v \in \bag(t)$, so $\mathsf{outH}(c_1) \cap \mathsf{outH}(c_2) \subseteq X_1 \cap \bag(t)$.
\end{proof}

\begin{lemma}
If $t$ is a $v$-top-node with child $c$ in $\Tc$, then either $\vbag(t) = \vbag(c)$ and $v \notin X_2$, or $\vbag(t) = \vbag(c) \cup \{v\}$ and $v \in X_2 \setminus \vbag(c)$.
\end{lemma}
\begin{proof}
We observe that $\vbag(t) = \vbag'(t)$ and $\vbag(c) = \vbag'(c)$, so it suffices to argue the conclusion for $\vbag'$.

As $\bag(t) = \bag(c)$ and $\subtree_{\Tc}(t) = \subtree_{\Tc}(c)$, we have that $(\bag(t) \cap X_1) \cup \mathsf{outH}(t) = (\bag(c) \cap X_1) \cup \mathsf{outH}(c)$.
The conclusion then follows from the facts that $\topv_{\Tc}(t) = \{v\}$ and $\topv_{\Tc}(c) = \emptyset$.
\end{proof}

\begin{lemma}
If $t$ is the root of $T$, then $\vbag(t) = \emptyset$.
\end{lemma}
\begin{proof}
It follows directly from the definition that $\vbag'(t) = \emptyset$.
We observe that also the edits to forget vertices one by one do not affect to this, because $\Tc$ is $\ell$-supernice and thus below $t$ there is a path of $\ell+1$ nodes $t'$, for each of which it holds that $\vbag'(t') = \emptyset$.
\end{proof}

This finishes the proof of \Cref{lem:virtdecomp}.

\section{Dynamic programming}
\label{sec:dynprog}
In this section, we turn the structural observations of \Cref{sec:innertree} into a dynamic programming algorithm.
In particular, this section is dedicated to the proof of the following lemma.

\begin{restatable}{lemma}{maindplemma}
\label{lem:maindplemma}
There is an algorithm that, given a graph $G$, a weight-function $\we \colon V(G) \to \R$, a tree decomposition $\Tc = (T,\bag)$ of $G$ with $\mu(\Tc) \le k$, a $\CMSO_2$-sentence $\Phi$, an integer $w$, and for each $t \in V(T)$ a family of sets $\fami(t)$, in time $f(k,w,|\Phi|) \cdot n^{\OO(k w^2)} \cdot |V(T)|^{\OO(1)} \cdot \sum_{t \in V(T)} |\fami(t)|^{\OO(1)}$, where $f$ is a computable function, finds a set $X \subseteq V(G)$ such that
\begin{itemize}
\item $\tw(G[X]) \le w$,
\item $G[X]$ satisfies $\Phi$,
\item $X \cap \bag(t) \in \fami(t)$ for all $t \in V(T)$, and
\item among all such sets, $X$ maximizes $\we(X)$,
\end{itemize}
or reports that no such set $X$ exists.
\end{restatable}

We prove \Cref{lem:maindplemma} by first describing the dynamic programming states in \Cref{subsec:dptables}, then the transitions in \Cref{subsec:dptrans}, and then putting things together in \Cref{subsec:dpprooffinal}.

\subsection{Dynamic programming states}
\label{subsec:dptables}
Let us start by describing the dynamic programming states.
For this, we introduce several technical definitions.
We use a parameter $\ell$, which will be set to the parameter $\ell$ of \Cref{lem:virtdecomp}.
We assume that we work with an $\ell$-supernice tree decomposition.

\paragraph{Inner tree decompositions.}
Let $\Tc = (T,\bag)$ be an $\ell$-supernice tree decomposition of $G$, $t \in V(T)$, and $X \subseteq \subtree_{\Tc}(t)$.
Denote by $\Tc_t = (T_t,\bag_t)$ the tree decomposition of $\subtree_{\Tc}(t)$ obtained by letting $T_t$ be the subtree of $T$ rooted at $t$ and $\bag_t$ the restriction of $\bag$ to $V(T_t)$.
Let $(X_1, X_2)$ be a pair of disjoint subsets of $V(G)$.
We say that a tree decomposition of form $(T_t,\vbag_t)$ is a \emph{$(t,B_1,B_2)$-inner tree decomposition} of $(X_1, X_2)$ if
\begin{enumerate}
\item $(T_t, \vbag_t)$ is a tree decomposition of the graph $G[X_1 \cup (X_2 \setminus \bag(t)) \cup B_2]$ of width at most $\ell$,
\item the bipartition $(X_1, (X_2 \setminus \bag(t)) \cup B_2)$ of $X_1 \cup (X_2 \setminus \bag(t)) \cup B_2$ and the tree decomposition $(T_t, \vbag_t)$ satisfy the properties of \Cref{lem:virtdecomp:prop1,lem:virtdecomp:prop2,lem:virtdecomp:prop3,lem:virtdecomp:prop4} of \Cref{lem:virtdecomp} for all nodes of $T_t$,
\item $\vbag_t(t) \cap X_1 = B_1$,
\item $\vbag_t(t) \cap X_2 = B_2$, and
\item for all nodes $x \in V(T_t)$, $N(X_2 \cap \bag(x)) \cap X_1 \cap \subtree_{\Tc}(x) \subseteq \vbag_t(x)$.\label{enum:definnertd:nbitem}
\end{enumerate}

Note that these conditions can be satisfied only when 

\begin{enumerate}
\item $\vbag_t(t) = B_1 \cup B_2$ and $|B_1 \cup B_2| \le \ell+1$,
\item $X_1 \cap \bag(t) \subseteq B_1 \subseteq X_1$,
\item $B_2 \subseteq X_2 \cap \bag(t)$ and $|B_2| \le 1$, and
\item $X_2$ is an independent set.
\end{enumerate}

\paragraph{DP-tuples.}
Let $\Tc = (T,\bag)$ be an $\ell$-supernice tree decomposition of $G$ and $\fami \colon V(T) \to 2^{V(G)}$ assign each node $t \in V(T)$ a set $\fami(t)$ of subsets of $\bag(t)$.
For a node $t \in V(T)$, we say that a set $X \subseteq \subtree_{\Tc}(t)$ is \emph{$t$-admissible} if for all descendants $t'$ of $t$ it holds that $X \cap \bag(t') \in \fami(t')$.
Let also $\autom$ be a tree decomposition automaton of width $\ell$ and state set $Q$.

As the dynamic programming table, we consider tuples of form $(t, X_t, B_1, B_2, q)$, where 
\begin{enumerate}
\item $t \in V(T)$, 
\item $X_t \in \fami(t)$, 
\item $B_1 \subseteq \subtree_{\Tc}(t)$,
\item $B_2 \subseteq X_t \setminus B_1$, $|B_2| \le 1$, $|B_1 \cup B_2| \le \ell+1$, $(B_1 \cup B_2) \cap \bag(t) \subseteq X_t$, and
\item $q \in Q$.
\end{enumerate}

We call such tuples \emph{DP-tuples}.
Note that there are at most $\sum_{t \in V(T)} |\fami(t)| \cdot n^{\OO(\ell)} \cdot |Q|$ DP-tuples, and given $\fami(t)$ for each $t$, they can be enumerated in such running time.

\paragraph{Pairs fitting DP-tuples.}
We say that a pair $(X_1,X_2)$ of disjoint subsets of $V(G)$ \emph{fits} to a DP-tuple $(t, X_t, B_1, B_2, q)$ if
\begin{enumerate}
\item $X_1 \cup X_2 \subseteq \subtree_{\Tc}(t)$,
\item $(X_1 \cup X_2) \cap \bag(t) = X_t$,
\item $X_1 \cup X_2$ is $t$-admissible,
\item $X_1 \cap \bag(t) = B_1 \cap \bag(t)$, (implying $X_2 \cap \bag(t) = X_t \setminus B_1$),
\item $X_2$ is an independent set,\label{enum:fits:isis}
\item $B_1 \subseteq X_1$ and $B_2 \subseteq X_2$,
\item there is a $(t, B_1, B_2)$-inner tree decomposition $(T_t, \vbag_t)$ of $(X_1, (X_2 \setminus X_t) \cup B_2)$, and\label{enum:fits:innerdecomp}
\item the state of $\autom$ on $(T_t,\vbag_t)$ is $q$.\label{enum:fits:innerdecompdp}
\end{enumerate}

We call a DP-tuple $(t, X_t, B_1, B_2, q)$ \emph{invalid} if either
\begin{enumerate}
\item $X_t \setminus B_1$ is not an independent set, or
\item $B_2 \neq \emptyset$ and $t$ is not a $v$-top-node for $\{v\} = B_2$.
\end{enumerate}
Note that no pair $(X_1,X_2)$ fits to an invalid DP-tuple.
If a DP-tuple is not invalid, it is \emph{valid}.

For a node $t \in V(T)$, a \emph{$t$-table} stores for each DP-tuple of form $\tau = (t, \cdot, \cdot, \cdot, \cdot)$ either $\bot$, if no pair $(X_1,X_2)$ fits to $\tau$, or a pair $(X^\tau_1, X^\tau_2)$ that fits to $\tau$ and maximizes $\we(X_1 \cup X_2)$ among all pairs $(X_1,X_2)$ that fit to $\tau$.

\subsection{Dynamic programming transitions}
\label{subsec:dptrans}
We then describe the dynamic programming transitions.
We retain all assumptions and definitions from the previous subsection, in particular, that $\Tc = (T,\bag)$ is an $\ell$-supernice tree decomposition.

\begin{lemma}[Introduce]
\label{lem:dptrans-intro}
Let $t \in V(T)$ be a $v$-introduce-node with child $c$.
Given a $c$-table, a $t$-table can be computed in time $|\fami(c)| \cdot |\fami(t)| \cdot n^{\OO(\ell)} \cdot |Q|^2 \cdot \evaltime(\autom)$.
\end{lemma}
\begin{proof}
We iterate over all valid DP-tuples of the form $\tau = (t, X_t, B_1, \emptyset, q)$.
Consider a pair $(X_1,X_2)$ that fits to $\tau$, and let $\Tc_t = (T_t, \vbag_t)$ be a $(t,B_1,\emptyset)$-inner tree decomposition of $(X_1, X_2 \setminus X_t)$ of width at most $\ell$ such that the state of $\autom$ on $\Tc_t$ is $q$.

\paragraph{Case 1: $v \in X_t$ and $v \in B_1$.}
Suppose that $v \in X_t$ and $v \in B_1$.
By \Cref{lem:virtdecomp}, it must be that $\vbag_t(t) = \vbag_t(c) \cup \{v\}$, and $v \notin \vbag_t(c)$.
Therefore, $(X_1 \setminus \{v\}, X_2)$ fits to a DP-tuple of form $(c, X_t \setminus \{v\}, B_1 \setminus \{v\}, \emptyset, q')$, where $q' \in Q$ such that $\transit^1_{\autom}(q', B_1 \setminus \{v\}, B_1) = q$.
Furthermore, we argue that for any pair $(X'_1, X'_2)$ that fits to such a DP-tuple, the pair $(X'_1 \cup \{v\}, X'_2)$ fits to $\tau$.
For this, the non-trivial thing to argue is that any $(c, B_1 \setminus \{v\}, \emptyset)$-inner tree decomposition of $(X'_1, X'_2)$ can be lifted to a $(t, B_1, \emptyset)$-inner tree decomposition of $(X'_1 \cup \{v\}, X'_2)$ by adding a root bag $B_1$.

Denote by $\Tc' = (T',\vbag')$ a tree decomposition obtained like that.
The fact that $\Tc'$ is a tree decomposition of $G[X'_1 \cup \{v\} \cup (X_2 \setminus X_t)]$ follows from the fact that because $t$ is a $v$-introduce-node, all neighbors of $v$ in $X'_1 \cup X'_2 \subseteq \subtree_{\Tc}(t)$ are in $X_t$, and therefore all of its neighbors in $X'_1 \cup (X'_2 \setminus X_t)$ are in $B_1$.
The facts that $(X'_1 \cup \{v\}, X'_2 \setminus X_t)$ and $\Tc'$ satisfy the properties of \Cref{lem:virtdecomp} are clear, as well as that $B_1 \cap (X'_1 \cup \{v\}) = B_1$ and that $B_1 \cap X_2 = B_2 = \emptyset$.
It remains to argue that for all nodes $x \in V(T')$, $N(X'_2 \cap \bag(x)) \cap (X'_1 \cup \{v\}) \cap \subtree_{\Tc}(x) \subseteq \vbag'(x)$.
The addition of $v$ into $X'_1$ affects this only for the node $x = t$, as that is the only node of $T'$ for which $v \in \subtree_{\Tc}(x)$.
For the node $x = t$ this holds because $X'_2 \cap \bag(t) = X'_2 \cap \bag(c)$ and $(X'_1 \cup \{v\}) \cap \subtree_{\Tc}(t) = (X'_1 \cap \subtree_{\Tc}(c)) \cup \{v\}$, but $\vbag'(t) = \vbag'(c) \cup \{v\}$.

Therefore, we can store for $\tau$ such pair $(X'_1 \cup \{v\}, X'_2)$ maximizing the weight.

\paragraph{Case 2: $v \in X_t$ and $v \notin B_1$.}
Suppose that $v \in X_t$ and $v \notin B_1$.
By \Cref{lem:virtdecomp}, it must be that $\vbag_t(t) = \vbag_t(c)$, and $v \notin \vbag_t(c)$.
Therefore, $(X_1, X_2 \setminus \{v\})$ fits to a DP-tuple of form $(c, X_t \setminus \{v\}, B_1, \emptyset, q')$, where $q' \in Q$ such that $\transit^1_{\autom}(q', B_1, B_1) = q$.
Furthermore, we argue that for any pair $(X'_1, X'_2)$ that fits to such a DP-tuple, the pair $(X'_1, X'_2 \cup \{v\})$ fits to $\tau$.

For this, the non-trivial things to argue are that (1) $X'_2 \cup \{v\}$ is an independent set and (2) any $(c, B_1, \emptyset)$-inner tree decomposition of $(X'_1, X'_2)$ can be lifted to a $(t, B_1, \emptyset)$-inner tree decomposition of $(X'_1, X'_2 \cup \{v\})$ by adding a root bag $B_1$.

For (1), we have that as $X'_2 \subseteq \subtree_{\Tc}(t)$, all neighbors of $v$ in $X'_2$ would be in $\bag(t)$, and therefore that $X'_2 \cup \{v\}$ is an independent set follows from the fact that $\tau$ is valid.
For (2), it suffices to argue that if $\Tc' = (T', \vbag')$ is a tree decomposition of $G[X'_1 \cup (X'_2 \setminus X_t)]$ constructed like that, then for all of its nodes $x \in V(T')$ it holds that $N((X'_2 \cup \{v\}) \cap \bag(x)) \cap X'_1 \cap \subtree_{\Tc}(x) \subseteq \vbag'(x)$.
This follows from the fact that $x = t$ is the only node in $T'$ with $v \in \bag(x)$, and the condition holds for $x = t$ because all neighbors of $v$ in $X'_1$ that are in $\subtree_{\Tc}(t)$ must be in $\bag(t)$, and therefore in $B_1$.

Therefore, we can store for $\tau$ such pair $(X'_1, X'_2 \cup \{v\})$ maximizing the weight.

\paragraph{Case 3: $v \notin X_t$.}
Suppose that $v \notin X_t$.
By \Cref{lem:virtdecomp}, it must be that $\vbag_t(t) = \vbag_t(c)$.
Therefore, $(X_1,X_2)$ fits to a DP-tuple of form $(c, X_t, B_1, \emptyset, q')$, where $q' \in Q$ such that $\transit^1_{\autom}(q', B_1, B_1) = q$.
Furthermore, we observe that any pair $(X_1', X_2')$ that fits to such a DP-tuple also fits to $\tau$, so we can store for $\tau$ the maximum-weight pair that fits to such a DP-tuple.
\end{proof}

\begin{lemma}[Forget]
\label{lem:dptrans-forget}
Let $t \in V(T)$ be a $v$-forget-node with child $c$.
Given a $c$-table, a $t$-table can be computed in time $|\fami(c)| \cdot |\fami(t)| \cdot n^{\OO(\ell)} \cdot |Q|^2 \cdot \evaltime(\autom)$.
\end{lemma}
\begin{proof}
We iterate over all valid DP-tuples of the form $\tau = (t, X_t, B_1, \emptyset, q)$.
Consider a pair $(X_1,X_2)$ that fits to $\tau$, and let $\Tc_t = (T_t, \vbag_t)$ be a $(t,B_1,\emptyset)$-inner tree decomposition of $(X_1, X_2 \setminus X_t)$ of width at most $\ell$ such that the state of $\autom$ on $\Tc_t$ is $q$.

\paragraph{Case 1: $v \notin B_1$.}
Suppose that $v \notin B_1$.
By \Cref{lem:virtdecomp}, it must be that either (1) $v \notin (X_1 \cup X_2)$ and $\vbag_t(t) = \vbag_t(c)$, or (2) $v \in X_2$, $v \in \vbag_t(c)$, and $\vbag_t(t) = \vbag_t(c) \setminus \{v\}$.

In (1), we have that $(X_1, X_2)$ fits to some DP-tuple of form $(c, X_t, B_1, \emptyset, q')$, where $q' \in Q$ such that $\sigma^1_{\autom}(q', B_1, B_1) = q$, and any set $(X'_1,X'_2)$ that fits to such a DP-tuple also fits to $\tau$.

In (2), we have that $(X_1,X_2)$ fits to some DP-tuple of form $(c, X_t \cup \{v\}, B_1, \{v\}, q')$, where $q' \in Q$ such that $\sigma^1_{\autom}(q', B_1 \cup \{v\}, B_1) = q$, and any pair $(X'_1, X'_2)$ that fits to such a DP-tuple also fits to $\tau$.

Therefore, we can store for $\tau$ the maximum-weight pair $(X'_1,X'_2)$ that fits to either (1) or (2).

\paragraph{Case 2: $v \in B_1$.}
Suppose that $v \in B_1$.
By \Cref{lem:virtdecomp} it must be that $\vbag_t(t) = \vbag_t(c)$.
Therefore, $(X_1,X_2)$ fits to some DP-tuple of form $(c, X_t \cup \{v\}, B_1, \emptyset, q')$, where $q' \in Q$ such that $\sigma^1_{\autom}(q', B_1, B_1) = q$.
Furthermore, we observe that any pair $(X'_1, X'_2)$ that fits to such a DP-tuple also fits to $\tau$, so we can store for $\tau$ the maximum-weight set that fits to such a DP-tuple.
\end{proof}

\begin{lemma}[Join]
\label{lem:dptrans-join}
Let $t \in V(T)$ be a join-node with children $c_1,c_2$.
Given a $c_1$-table and a $c_2$-table, a $t$-table can be computed in time $|\fami(c_1)| \cdot |\fami(c_2)| \cdot |\fami(t)| \cdot n^{\OO(\ell)} \cdot |Q|^3 \cdot \evaltime(\autom)$.
\end{lemma}
\begin{proof}
We iterate over all valid DP-tuples of the form $\tau = (t, X_t, B_1, \emptyset, q)$.
Consider a pair $(X_1,X_2)$ that fits to $\tau$, and let $\Tc_t = (T_t, \vbag_t)$ be a $(t,B_1,\emptyset)$-inner tree decomposition of $(X_1, X_2 \setminus X_t)$ of width at most $\ell$ such that the state of $\autom$ on $\Tc_t$ is $q$.

By \Cref{lem:virtdecomp}, it must be that $\vbag_t(t) = \vbag_t(c_1) \cup \vbag_t(c_2)$ and $\vbag_t(c_1) \cap \vbag_t(c_2) = \vbag_t(t) \cap \bag(t)$.
As $\vbag_t(t) = B_1$, $\vbag_t(c_1) \subseteq \subtree_{\Tc}(c_1)$, $\vbag_t(c_2) \subseteq \subtree_{\Tc}(c_2)$, and $\subtree_{\Tc}(c_1) \cap \subtree_{\Tc}(c_2) = \bag(t)$, it follows that $\vbag_t(c_1) = B_1 \cap \subtree_{\Tc}(c_1)$ and $\vbag_t(c_2) = B_1 \cap \subtree_{\Tc}(c_2)$.
Thus, denote $B_1^1 = B_1 \cap \subtree_{\Tc}(c_1)$ and $B_1^2 = B_1 \cap \subtree_{\Tc}(c_2)$.

Therefore, $(X_1 \cap \subtree_{\Tc}(c_1), X_2 \cap \subtree_{\Tc}(c_1))$ fits to a DP-tuple of form $(c_1, X_t, B_1^1, \emptyset, q'_1)$ and $(X_1 \cap \subtree_{\Tc}(c_2), X_2 \cap \subtree_{\Tc}(c_2))$ to a DP-tuple of form $(c_2, X_t, B_1^2, \emptyset, q'_2)$, such that $\transit^2_{\autom}(q'_1, q'_2, B_1^1, B_1^2, B_1) = q$.
We argue also the following.

\begin{claim}
For any such pair of pairs $(X^1_1, X^1_2)$ and $(X^2_1, X^2_2)$, the pair $(X^1_1 \cup X^2_1, X^1_2 \cup X^2_2)$ fits to $\tau$.
\end{claim}
\begin{claimproof}
Most of the conditions are trivial to check, so let us check \Cref{enum:fits:isis,enum:fits:innerdecomp,enum:fits:innerdecompdp}, i.e., that (1) $X^1_2 \cup X^2_2$ is an independent set, and (2) if we take for each $i \in [2]$ a $(c_i, B_1^i, \emptyset)$-inner tree decomposition $\Tc_i$ of $(X^i_1, X^i_2 \setminus X_t)$, and combine $\Tc_1$ with $\Tc_2$ by connecting them from their root-bags to a new root-bag $B_1$, then we obtain a $(t,B_1,\emptyset)$-inner tree decomposition $\Tc' = (T', \vbag')$ of $(X^1_1 \cup X^2_1, (X^1_2 \cup X^2_2) \setminus X_t)$.

For (1), it suffices to recall that there are no edges between $\subtree_{\Tc}(c_1) \setminus \bag(t)$ and $\subtree_{\Tc}(c_2) \setminus \bag(t)$, so for any edge $uv$ with $u,v \in \subtree_{\Tc}(t)$, either $\{u,v\} \subseteq \subtree_{\Tc}(c_1)$ or $\{u,v\} \subseteq \subtree_{\Tc}(c_2)$.

For (2), we note that as $B_1 = B_1^1 \cup B_1^2$, such a tree decomposition clearly satisfies the vertex-condition and the connectedness-condition.
We then prove the edge-condition.

Consider an edge $uv$ with both ends in $X^1_1 \cup X^2_1 \cup ((X^1_2 \cup X^2_2) \setminus X_t)$.
By the above argument, either $\{u,v\} \subseteq \subtree_{\Tc}(c_1)$ or $\{u,v\} \subseteq \subtree_{\Tc}(c_2)$, so assume first $\{u,v\} \subseteq \subtree_{\Tc}(c_1)$.
However, now $\{u,v\} \subseteq X^1_1 \cup (X^1_2 \setminus X_t)$, so $u$ and $v$ share a bag already in $\Tc_1$.
The argument for the case when $\{u,v\} \subseteq \subtree_{\Tc}(c_2)$ is similar.

Finally, we should check that for all nodes $x \in V(T')$, $N((X^1_2 \cup X^2_2) \cap \bag(x)) \cap (X^1_1 \cup X^2_1) \cap \subtree_{\Tc}(x) \subseteq \vbag'(x)$.
For all descendants $y$ of $c_i$ it holds that $(X^1_2 \cup X^2_2) \cap \bag(y) = X^i_2 \cap \bag(y)$ and that $(X^1_1 \cup X^2_1) \cap \subtree_{\Tc}(y) = X^i_1 \cap \subtree_{\Tc}(y)$.
Therefore, it suffices to check this for $x = t$.
For $x = t$, this follows from the facts that (1) for every edge $uv$ with $u \in X^1_1 \cup X^2_1$ and $v \in X^1_2 \cup X^2_2$, either $\{u,v\} \subseteq X^1_1 \cup X^1_2$, or $\{u,v\} \subseteq X^2_1 \cup X^2_2$, and (2) $B_1 = B_1^1 \cup B_1^2$.
\end{claimproof}

Therefore, we can store for $\tau$ the maximum-weight pair that is obtained as a union of two pairs $(X^1_1, X^1_2)$ and $(X^2_1, X^2_2)$ that fit to such DP-tuples $(c_1, X_t, B_1^1, \emptyset, q'_1)$ and $(c_2, X_t, B_1^2, \emptyset, q'_2)$, respectively.
\end{proof}

\begin{lemma}[Neutral]
\label{lem:dptrans-neutral}
Let $t \in V(T)$ be a neutral-node with child $c$.
Given a $c$-table, a $t$-table can be computed in time $|\fami(c)| \cdot |\fami(t)| \cdot n^{\OO(\ell)} \cdot |Q|^2 \cdot \evaltime(\autom)$.
\end{lemma}
\begin{proof}
We iterate over all valid DP-tuples of the form $\tau = (t, X_t, B_1, \emptyset, q)$.
Consider a pair $(X_1,X_2)$ that fits to $\tau$, and let $\Tc_t = (T_t, \vbag_t)$ be a $(t,B_1,\emptyset)$-inner tree decomposition of $(X_1, X_2 \setminus X_t)$ of width at most $\ell$ such that the state of $\autom$ on $\Tc_t$ is $q$.
By \Cref{lem:virtdecomp}, either (1) $\vbag_t(c) = \vbag_t(t)$ or (2) $\vbag_t(c) = \vbag_t(t) \cup \{v\}$ for some $v \in X_1 \setminus \vbag_t(t)$.

\paragraph{Case (1).} In case (1), we have that $(X_1,X_2)$ fits to a DP-tuple of form $(c, X_t, B_1, \emptyset, q')$, where $q' \in Q$ such that $\sigma^1_{\autom}(q', B_1, B_1) = q$.
Furthermore, we observe that any pair $(X_1',X_2')$ that fits to such a DP-tuple also fits to $\tau$.

\paragraph{Case (2).}
For case (2), we observe that it must be that $v \in \subtree_{\Tc}(t) \setminus (\bag(t) \cup B_1)$.
Furthermore, because of \Cref{enum:definnertd:nbitem} of the definition of $(B_1,B_2)$-inner tree decomposition, it must be that $N(v) \cap (X_t \setminus B_1) = \emptyset$.
Assume that $v$ satisfies the above conditions.

Now, we have that $(X_1,X_2)$ fits to a DP-tuple of form $(c, X_t, B_1 \cup \{v\}, \emptyset, q')$, where $q' \in Q$ such that $\sigma^1_{\autom}(q', B_1 \cup \{v\}, B_1) = q$, and $v \in \subtree_{\Tc}(c) \setminus \bag(t)$ such that $v$ is not adjacent to any vertex in $X_t \setminus B_1$.
We argue that any pair $(X_1', X_2')$ that fits to such a DP-tuple also fits to $\tau$.
For this, most of the conditions are easy to check, but we should verify that any $(c, B_1 \cup \{v\}, \emptyset)$-inner tree decomposition $\Tc_c$ of $(X_1', X_2')$ can be lifted to a $(t, B_1, \emptyset)$-inner tree decomposition $\Tc' = (T', \vbag')$ of $(X_1', X_2')$ by adding a root-bag $B_1$.
Clearly, $\Tc'$ is a tree decomposition of $G[X_1' \cup (X'_2 \setminus \bag(t))]$, and satisfies with $(X'_1, X'_2)$ the conditions of \Cref{lem:virtdecomp}.
It remains to check that for all nodes $x \in V(T')$, $N(X'_2 \cap \bag(x)) \cap X'_1 \cap \subtree_{\Tc}(x) \subseteq \vbag'(x)$.
This clearly holds for $x \neq t$, so consider $x = t$.
As $\vbag'(t) = \vbag'(c) \setminus \{v\}$ and $\bag(t) = \bag(c)$, it suffices to consider only the vertex $v$ from $X'_1$.
For $v$, this holds by the choice of $v$, in particular, because $v$ is not adjacent to any vertex in $X_t \setminus B_1 = X'_2 \cap \bag(t)$.

Therefore, we can store for $\tau$ the maximum-weight pair that either fits to (1), or fits to (2) for some choice of $v$ that satisfies the aforementioned conditions.
\end{proof}

\begin{lemma}[Top]
\label{lem:dptrans-top}
Let $t \in V(T)$ be a $v$-top-node with child $c$.
Given a $c$-table, a $t$-table can be computed in time $|\fami(c)| \cdot |\fami(t)| \cdot n^{\OO(\ell)} \cdot |Q|^2 \cdot \evaltime(\autom)$.
\end{lemma}
\begin{proof}
We iterate over all valid DP-tuples of the form $\tau = (t, X_t, B_1, B_2, q)$.
Recall that thus, $B_2 = \{v\}$ or $B_2 = \emptyset$.
Consider a pair $(X_1,X_2)$ that fits to $\tau$, and let $\Tc_t = (T_t, \vbag_t)$ be a $(t,B_1,B_2)$-inner tree decomposition of $(X_1, X_2 \setminus X_t)$ of width at most $\ell$ such that the state of $\autom$ on $\Tc_t$ is $q$.

\paragraph{Case 1: $B_2 = \{v\}$.}
By \Cref{lem:virtdecomp}, in this case it holds that $\vbag_t(c) = \vbag_t(t) \setminus \{v\}$, and furthermore $\vbag_t(t)$ is the only bag of $(T_t, \vbag_t)$ containing $v$.

We have that $(X_1,X_2)$ fits to a DP-tuple of form $(c, X_t, B_1, \emptyset, q')$, where $q' \in Q$ such that $\sigma^1_{\autom}(q', B_1, B_1 \cup \{v\}) = q$.
We claim that every pair $(X_1',X_2')$ that fits to such DP-tuple also fits to $\tau$.
The non-trivial thing to argue is that any $(c, B_1, \emptyset)$-inner tree decomposition $\Tc_c = (T_c, \vbag_c)$ of $(X_1', X_2')$ can be lifted to a $(t, B_1, \{v\})$-inner tree decomposition $\Tc'$ of $(X_1', X_2')$ by adding a root bag $B_1 \cup \{v\}$.

Let us check that $\Tc'$ is indeed a tree decomposition of $G[X_1' \cup (X_2' \setminus X_t) \cup \{v\}]$.
The vertex-condition is satisfied because all vertices of $X_1' \cup (X_2' \setminus X_t) \cup \{v\}$ except $v$ occur in $\Tc_c$, while $v$ occurs in the new root bag.
The connectedness condition follows from the fact that $\Tc_c$ is a tree decomposition of $G[X_1' \cup (X_2' \setminus X_t)]$ with root bag $B_1$, and in particular, $v$ does not occur in any of its bags.
For the edge-condition consider an edge $uv$ between $v$ and a vertex $u \in X_1' \cup (X_2' \setminus X_t)$.
As $v \in X_2'$ and $X_2'$ is an independent set, we have that $u$ must be in $X_1'$, and therefore also in $X_1' \cap \subtree_{\Tc}(c)$.
As $v \in X_2' \cap \bag(c)$, it follows from \Cref{enum:definnertd:nbitem} of the definition of $(t, B_1, B_2)$-inner tree decomposition that $u \in \vbag_c(c)$, and therefore $u \in B_1$.

The other conditions in the definition of $(t, B_1, B_2)$-inner tree decomposition are easy to check.
Therefore, we can store for $\tau$ the maximum-weight pair that fits to such a DP-tuple.

\paragraph{Case 2: $B_2 = \emptyset$.}
By \Cref{lem:virtdecomp}, in this case it holds that $\vbag_t(c) = \vbag_t(c)$.
It follows that $(X_1,X_2)$ fits to some DP-tuple of form $(c, X_t, B_1, \emptyset, q')$, where $q' \in Q$ such that $\sigma^1_{\autom}(q', B_1, B_1) = q$.
Furthermore, we observe that any pair $(X_1',X_2')$ that fits to such a DP-tuple also fits to $\tau$, so we can store for $\tau$ the maximum-weight pair that fits to such a DP-tuple.
\end{proof}

\begin{lemma}[Initial]
\label{lem:dptrans-initial}
Let $t \in V(T)$ be an initial-node.
A $t$-table can be computed in time $n^{\OO(1)}$.
\end{lemma}
\begin{proof}
The only DP-tuple to consider is $\tau = (t, \emptyset, \emptyset, \emptyset, \sigma_{\autom}^0(\emptyset))$.
If $\emptyset \in \fami(t)$, then we store the pair $(\emptyset, \emptyset)$, otherwise we store $\bot$.
\end{proof}

\subsection{Proof of \Cref{lem:maindplemma}}
\label{subsec:dpprooffinal}
We now finish the proof of \Cref{lem:maindplemma} by putting the ingredients from \Cref{sec:innertree,subsec:dptables,subsec:dptrans} together.
Let us re-state it for convenience.

\maindplemma*
\begin{proof}
We select $\ell$ as the integer $\ell \le \OO(kw^2)$ given by \Cref{lem:virtdecomp}.
Then, we use the algorithm of \Cref{lem:tdintosupernice} to obtain, in time $(|V(T)|+n)^{\OO(1)}$, from $\Tc$ an $\ell$-supernice tree decomposition $\Tc' = (T', \bag')$ of $G$ such that 
\begin{itemize}
\item for each $t' \in V(T')$, there exists $t \in V(T)$ with $\bag'(t') \subseteq \bag(t)$, and
\item for each $t \in V(T)$, there exists $t' \in V(T')$ with $\bag'(t') = \bag(t)$.
\end{itemize}
We construct a set $\fami'(t')$ for each $t' \in V(T')$ by taking the intersection of the set families $\{X_t \cap \bag'(t') \mid X_t \in \fami(t)\}$ over all $t \in V(T)$ with $\bag'(t') \subseteq \bag(t)$.
Observe that now, a set $X \subseteq V(G)$ satisfies that $X \cap \bag(t) \in \fami(t)$ for all $t \in V(T)$ if and only if it satisfies that $X \cap \bag'(t') \in \fami'(t')$ for all $t' \in V(T')$.
These sets $\fami'$ can be constructed in time $(|V(T)| + n + \sum_{t \in V(T)} |\fami(t)|)^{\OO(1)} = |V(T)|^{\OO(1)} \cdot n^{\OO(1)} \cdot \sum_{t \in V(T)} |\fami(t)|^{\OO(1)}$.

Next, we use \Cref{lem:cmsotwcheck,lem:cmsotreeautomaton} to construct a tree decomposition automaton $\autom = (\sigma^0, \sigma^1, \sigma^2, Q, F)$ of width $\ell$ so that
\begin{itemize}
\item the state of $\autom$ on a tree decomposition of a graph $G$ is in $F$ if and only if $G$ satisfies $\Phi$ and $\tw(G) \le w$, and
\item $|Q|, \evaltime(\autom) \le g(\ell, |\Phi|, w)$ for a computable function $g$.
\end{itemize}
This runs in time bounded by a computable function of $\ell$, $w$, and $|\Phi|$, and thus by a computable function of $k$, $w$, and $|\Phi|$.

We then apply the definitions of \Cref{subsec:dptables} with the decomposition $\Tc'$, the integer $\ell$, and the automaton $\autom$, and use the algorithms of \Cref{lem:dptrans-intro,lem:dptrans-forget,lem:dptrans-join,lem:dptrans-neutral,lem:dptrans-top,lem:dptrans-initial} to compute $t$-tables for all nodes $t$ of $T'$.
This runs in time $|V(T')| \cdot \left(\sum_{t \in V(T')} |\fami'(t)|\right)^{3} \cdot n^{\OO(\ell)} \cdot g(\ell, |\Phi|, w)^4 = f(k, w, |\Phi|) \cdot n^{\OO(kw^2)} \cdot |V(T)|^{\OO(1)} \cdot \sum_{t \in V(T)} |\fami(t)|^{\OO(1)}$ for a computable function $f$.

Let $r$ be the root of $T'$.
It remains to show that the $r$-table indeed stores an optimal solution.

\begin{claim}
\label{lem:maindplemma:claim1}
If a pair $(X_1, X_2)$ fits to a DP-tuple of form $(r, \emptyset, \emptyset, \emptyset, q)$, where $q \in F$, then 
\begin{itemize}
\item $\tw(G[X_1 \cup X_2]) \le w$,
\item $G[X_1 \cup X_2]$ satisfies $\Phi$, and
\item $X \cap \bag(t) \in \fami(t)$ for all $t \in V(T)$.
\end{itemize}
\end{claim}
\begin{claimproof}
By definitions, then there exists a tree decomposition $\Tc_r$ of $G[X_1 \cup X_2]$ such that the state of $\autom$ on $\Tc$ is $q$, which by the construction of $\autom$ implies that $\tw(G[X_1 \cup X_2]) \le w$ and $G[X_1 \cup X_2]$ satisfies $\Phi$.
Furthermore, then $X_1 \cup X_2$ is $r$-admissible, implying that $(X_1 \cup X_2) \cap \bag'(t) \in \fami'(t)$ for all $t \in V(T')$, implying that $(X_1 \cup X_2) \cap \bag(t) \in \fami(t)$ for all $t \in V(T)$.
\end{claimproof}

\begin{claim}
\label{lem:maindplemma:claim2}
If a set $X \subseteq V(G)$ satisfies that
\begin{itemize}
\item $\tw(G[X]) \le w$,
\item $G[X]$ satisfies $\Phi$, and
\item $X \cap \bag(t) \in \fami(t)$ for all $t \in V(T)$,
\end{itemize}
then there exists a bipartition $(X_1,X_2)$ of $X$ and $q \in F$ such that $(X_1, X_2)$ fits to the DP-tuple $(r, \emptyset, \emptyset, \emptyset, q)$.
\end{claim}
\begin{claimproof}
Let $(X_1,X_2)$ and $\hat{\Tc} = (T', \vbag)$ be a bipartition of $X$ and a tree decomposition of $G[X]$ given by \Cref{lem:virtdecomp}.
We observe (using \Cref{lem:virtdecomp-furtherconstraints}) that $\hat{\Tc}$ is indeed an $(r,\emptyset,\emptyset)$-inner tree decomposition of $(X_1,X_2)$, as defined in \Cref{subsec:dptables}.
Furthermore, because $\hat{\Tc}$ is a tree decomposition of $G[X]$, and $\tw(G[X]) \le w$ and $G[X]$ satisfies $\Phi$, the state of $\autom$ on $\hat{\Tc}$ is $q \in F$.
It follows that $(X_1, X_2)$ fits to the DP-tuple $(r, \emptyset, \emptyset, \emptyset, q)$.
\end{claimproof}

By \Cref{lem:maindplemma:claim1,lem:maindplemma:claim2}, it suffices then to iterate through DP-tuples of form $(r, \emptyset, \emptyset, \emptyset, q)$ for $q \in F$, and return the stored pair $(X_1,X_2)$ that maximizes $\we(X_1 \cup X_2)$, or if all of them contain $\bot$, return that no such set $X$ exist.
\end{proof}

\section{Proof of \Cref{the:main}}\label{sec:finalproof}
In this section we put together the results of the previous sections to prove \Cref{the:main}.

Before that, we need one more ingredient, which is an algorithm for computing a tree decomposition $\Tc$ with small $\mu(\Tc)$.
Such an algorithm was given originally by Yolov~\cite{DBLP:conf/soda/Yolov18}, and with an improved running time by Dallard, Fomin, Golovach, Korhonen, and Milanič~\cite{DBLP:conf/icalp/DallardFGKM24}.

\begin{lemma}[\cite{DBLP:conf/icalp/DallardFGKM24}]
\label{lem:computedecomp}
There is an algorithm that, given an $n$-vertex graph $G$ and an integer $k$, in time $2^{\OO(k^2)} n^{\OO(k)}$ either outputs a tree decomposition $\Tc$ of $G$ with $\mu(\Tc) \le 8k$, or concludes that $\treemu(G) > k$.
\end{lemma}

Now we can prove \Cref{the:main}.

\maintheorem*
\begin{proof}
We first apply the algorithm of \Cref{lem:computedecomp}.
If it returns that $\treemu(G) > k$, we can return that also.
Otherwise, we obtain a tree decomposition $\Tc = (T,\bag)$ of $G$ with $\mu(\Tc) \le 8k$.

We then apply the algorithm of \Cref{lem:bagsignatures} to compute for each $t \in V(T)$ a family $\fami(t)$ of subsets of $\bag(t)$, so that if $X \subseteq V(G)$ is the optimal solution to $(G,\we, \Phi, w)$, then $X \cap \bag(t) \in \fami(t)$ for all $t$, and $|\fami(t)| \le f_1(8k, w, |\Phi|) \cdot n^{\OO(kw)}$ for a computable function $f_1$.
This runs in time $f_1(8k, w, |\Phi|) \cdot |V(T)| \cdot n^{\OO(kw)} = f_1(8k, w, |\Phi|) \cdot n^{\OO(kw)}$.

We then apply the algorithms of \Cref{lem:maindplemma} with $G$, $\we$, $\Tc$, $\Phi$, $w$, the sets $\fami(t)$, to find a set $X \subseteq V(G)$ so that $\tw(G[X]) \le w$, $G[X]$ satisfies $\Phi$, $X \cap \bag(t) \in \fami(t)$ for all $t \in V(T)$, and among such sets, $X$ maximizes $\we(X)$, or conclude that no such set $X$ exists.
We return the conclusion returned by the algorithm of \Cref{lem:maindplemma}.
This runs in time $f_2(8k, w, |\Phi|) \cdot n^{\OO(kw^2)} \cdot |V(T)|^{\OO(1)} \cdot \sum_{t \in V(T)} |\fami(t)|^{\OO(1)} = f_2(8k, w, |\Phi|) \cdot n^{\OO(kw^2)} \cdot f_1(8k, w, |\Phi|)^{\OO(1)}$, for a computable function $f_2$, which is also a bound for the total running time (assuming $f_1(8k, w, |\Phi|) \ge 2^{\Omega(k^2)}$).
\end{proof}

We then prove \Cref{cor:kttfreealg} from \Cref{the:main} by using the following observation by Lima, Milanič, Muršič, Okrasa, Rzążewski, and Štorgel~\cite{DBLP:conf/esa/LimaMMORS24}.

\begin{lemma}[\cite{DBLP:conf/esa/LimaMMORS24}]
\label{lem:kttfreebdtw}
There is a (computable) function $f(k,t)$, so that every $K_{t,t}$-subgraph-free graph with $\treemu(G) \le k$ has treewidth at most $f(k,t)$.
\end{lemma}

\corkttfree*
\begin{proof}
We apply the algorithm of \Cref{the:main} with the treewidth bound obtained from \Cref{lem:kttfreebdtw}, and with $\Phi$ edited to include a clause that stating that $G[X]$ is $K_{t,t}$-subgraph-free (such a clause can be added while increasing the length of $\Phi$ by $\OO(t)$).
\end{proof}

\section{Conclusions}
\label{sec:conclusions}
We showed that \mwisbt is solvable in polynomial-time on graphs of bounded $\treemu$, resolving a conjecture of Lima et al.~\cite{DBLP:conf/esa/LimaMMORS24}.
Let us discuss in this section potential extensions of our main theorem and open problems.

First, we note that the proof \Cref{the:main} straightforwardly extends to the setting where the vertices and edges of the graph $G$ are labeled by a (constant number of) colors, and the formula $\Phi$ can talk about these colors.
In particular, this is the setting of $\CMSO_2$ on arbitrary binary structures.
With this, one could express, for example, the problem of finding a longest induced $s,t$-path for a given pair of vertices $s,t$.

There are also other generalizations of \Cref{the:main}, which we suspect could be proven by generalizations of our arguments.
For example, \Cref{the:main} does not capture the problem of $c$-colorability for constant $c$, which was shown by Yolov~\cite{DBLP:conf/soda/Yolov18} to be polynomial-time solvable on graphs of bounded $\treemu$.
A common generalization of this and \Cref{the:main} would be to ask to find a maximum set $X \subseteq V(G)$ that can be covered by (or partitioned into) $c$ sets $X_1,\ldots,X_c$, so that each $G[X_i]$ has treewidth at most $w_i$ and satisfies $\Phi_i$.
It would be interesting to know whether our arguments generalize to these problems; we suspect that the partitioning problem could be harder than the covering problem.

Another generalized formulation of \mwisbt was considered by Fomin, Todinca, and Villanger~\cite{FominTV15}.
They considered the problem of finding sets $X \subseteq F \subseteq V(G)$, so that $|X|$ is maximized, $G[F]$ has bounded treewidth, and the pair $(G[F], X)$ (consider this a graph labeled by an unary predicate) satisfies a $\CMSO_2$ sentence $\Phi$.
This captures, for example, the problem of packing the maximum number of disjoint mutually induced cycles.
We suspect that our techniques could extend to solve also this problem, although this is not straightforward.

A more interesting direction than generalizing \Cref{the:main} to capture more problems would be to generalize the parameter $\treemu$ itself.
For example, if instead of measuring the intersection of a bag with an induced matching (i.e., an induced collection of $K_2$), we would measure the intersection of a bag with an induced collection of $P_3$, or of cycles, would similar algorithmic results hold?
In particular, is \mwis polynomial-time solvable when given a tree decomposition with such properties?

\appendix

\bibliographystyle{alpha}
\bibliography{main}

\end{document}